\documentclass{fundam}

\usepackage{amssymb}
\usepackage{mathtools}
\usepackage[inline]{enumitem}

\usepackage{url} % takes care of hyperlinks, preferred over hyperref
\usepackage[hidelinks]{hyperref}

\let\iref\eqref

\DeclareMathOperator{\card}{card}

\DeclareMathOperator{\lhs}{lhs}
\DeclareMathOperator{\rhs}{rhs}
\DeclareMathOperator{\abs}{abs}

\newcommand{\strlen}[1]{\left\lvert#1\right\rvert}

\newcommand{\Ra}{\Rightarrow}

\newcommand{\lRa}[1]{\mathbin{\prescript{}{\mathit{left}_{#1}\csname prescriptcorrection\detokenize{#1}\endcsname}{\Ra}}}
\newcommand{\rRa}[1]{\mathbin{\prescript{}{\mathit{right}_{#1}\csname prescriptcorrection\detokenize{#1}\endcsname}{\Ra}}}
\makeatletter\@namedef{prescriptcorrection\detokenize{n}}{\mkern-10mu\relax}\makeatother

\newcommand{\singlfam}{sing}
\newcommand{\finlfam}{fin}
\newcommand{\reglfam}{reg}
\newcommand{\evenlfam}{even}
\newcommand{\oddlfam}{odd}

\newcommand{\fa}{FA}
\newcommand{\gfa}{GFA}

\newcommand{\twfa}{TW\fa}
\newcommand{\ietwsfa}{IETWS\fa}
\newcommand{\ietwgfa}{IETW\gfa}

\newcommand{\ling}{LG}
\newcommand{\eling}{E\ling}

\newcommand{\lfam}[2]{\prescript{\mathit{#1}}{\mathit{#2}}{\Phi}}

\let\originalleft\left
\let\originalright\right
\renewcommand{\left}{\mathopen{}\mathclose\bgroup\originalleft}
\renewcommand{\right}{\aftergroup\egroup\originalright}

\newcommand{\state}[1]{\left\langle#1\right\rangle}

\newcommand{\alt}{\mathit{alt}}
\newcommand{\even}{\mathit{even}}
\newcommand{\ieven}{{\mathit{init\mhyphen even}}}

\renewcommand{\iff}{\text{ iff }}
\renewcommand{\implies}{\text{ implies }}

\let\epsilon\varepsilon
\let\emptyset\varnothing

\newcommand{\lsh}{{\Lsh}}
\newcommand{\rsh}{{\Rsh}}

\newcommand{\stl}{\lsh}
\newcommand{\str}{\rsh}

\allowdisplaybreaks

\numberwithin{equation}{definition}

\mathchardef\mhyphen="2D

\title{Input-Erasing Two-Way Finite Automata}
\author{Alexander Meduna%\thanks{Address for correspondence: Department of Information Systems, Faculty of Information Technology, Brno University of Technology, Bo\v{z}et\v{e}chova 1/2, 612 00 Brno, Czech Republic}
, Dominik Nejedlý, Zbyněk Křivka\corresponding\\
Department of Information Systems\\
Faculty of Information Technology\\
Brno University of Technology\\
Bo\v{z}et\v{e}chova 1/2, 612 00 Brno, Czech Republic\\
meduna@fit.vut.cz, xnejed09@stud.fit.vut.cz, krivka@fit.vut.cz
}

\date{\today}

\runninghead{A. Meduna, D. Nejedlý, Z. Křivka}{Input-Erasing Two-Way Finite Automata}

\begin{document}

\maketitle
%\address{A. Meduna, Department of Information Systems, Faculty of Information Technology, Brno University of Technology, Bo\v{z}et\v{e}chova 1/2, 612 00 Brno, Czech Republic}

\begin{abstract}
    The present paper introduces and studies an alternative concept of two-way finite automata called input-erasing two-way finite automata. Like the original model, these new automata can also move the reading head freely left or right on the input tape. However, each time they read a~symbol, they also erase it from the tape. The paper demonstrates that these automata define precisely the family of linear languages and are thus strictly stronger than the original ones. Furthermore, it introduces a variety of restrictions placed upon these automata and the way they work and investigates the effect of these restrictions on their acceptance power. In particular, it explores the mutual relations of language families resulting from some of these restrictions and shows that some of them reduce the power of these automata to that of even linear grammars or even ordinary finite automata.
\end{abstract}


\section{Introduction}\label{sec:introduction}

Finite automata, introduced more than eight decades ago in \cite{A_LOGICAL_CALCULUS_MCCULLOCH_1943}, have always fulfilled a crucially important role in computer science both in theory and in practice. It thus comes as no surprise that the theory of computation has defined a great variety of these automata in order to provide every computer science area with the version that fits its needs as optimally as possible. Two-way finite automata, independently introduced in \cite{FINITE_AUTOMATA_DECISION_PROBLEMS_RABIN_1959} and \cite{TWO_WAY_FINITE_AUTOMATA_SHEPHERDSON_1959}, represent significant versions of this kind, which have been constantly and intensively investigated since their introduction from various angles. First of all, in terms of these automata, the theory of computation has studied most of its classical topics, such as nondeterminism, time and space complexity, or purely mathematical properties (see \cite{REMOVING_BIDRECTIONALITY_KAPOUTSIS_2005, NONDETERMINISM_AND_THE_SIZE_OF_TWO_WAY_SAKODA_1978, TWO_WAY_AUTOMATA_LOG_SPACE_KAPOUTSIS_2014, TWO_WAY_OLD_AND_RECENT_PIGHIZZINI_2012, FINITE_AUTOMATA_AND_UNARY_LANGS_CHROBAK_1986, TWO_WAY_DETERMINISCTIC_AUTOMATA_WITH_TWO_REVERSALS_BALCERAK_2010, POSITIONAL_SIMULATION_OF_TWO_WAY_AUTOMATA_BIRGET_1992, NONDET_DET_FOR_TWO_WAY_FINITE_AUTOMATA_HROMKOVIC_2003, TWO_WAY_AUTOMATA_AND_LENGTH_PRESERVING_BIRGET_1996, TWO_WAY_AUTOMATA_SIMULATIONS_AND_UNARY_LANGS_MEREGHETTI_2000, TWO_WAY_TO_UNAMBIGOUS_AUTOMATA_PETROV_2023, BASIC_TECHNIQUES_FOR_TWO_WAY_BIRGET_1987}). Furthermore, this theory has introduced various formal models closely related to two-way finite automata in many respects, such as their power or the way they work (see \cite{BIAUTOMATA_JIRASKOVA_KLIMA_2022, LINEAR_CONTEXT_FREE_LOUKA_2007}). The theory of computation has also defined and studied several new versions of these automata based upon concepts used in its latest investigation trends, such as the formalization of quantum or jumping computation (see \cite{TWO_TAPE_ZHENG_2011, TWO_WAY_FINITE_AUTOMATA_WITH_QUANTUM_AMBAINIS_2002, SOME_LANGS_RECOGNIZED_BY_TWO_WAY_ZHENG_2011, SUCCINCTNESS_TWO_WAY_PROB_QUANTUM_ABUZER_2010, SOME_OBSERVATIONS_ON_TWO_WAY_DAOWEN_2008, TWO_WAY_JUMPING_FAZEKAS_2021}). In addition, apart from the models mentioned above, many other versions of two-way finite automata have been introduced to formalize various features of computation in such terms as probability, alternation, and others (see \cite{PROB_TWO_WAY_FREIVALDS_1981, PROB_TWO_WAY_VARIANTS_RAVIKUMAR_2007, TWO_WAY_NON_UNIFORM_FREI_2021, TWO_WAY_ALTERNATION_KAPOUTSIS_2021, OBLIVIOUS_TWO_WAY_KUTRIB_2014}).

The present paper continues with this long-time vivid investigation trend by introducing other versions of two-way finite automata, which are, however, stronger than their originals. Indeed, these newly introduced versions characterize the linear language family, which properly contains the regular language family defined by two-way finite automata. Considering this increase in power, it is surprising that the fundamental idea underlying these new versions actually comes from the very original concept of one-way finite automata, which can read every input symbol only once. That is to say, once an input symbol is read, it is also erased, so it cannot be re-read again. To give a more detailed insight into these new versions as well as the way they work, we first informally recall the basic notion of a (one-way) finite automaton as well as that of its two-way variant while pointing out the features that have inspired the introduction of the new versions.

Conceptually, the notion of a (one-way) finite automaton $M_1$ consists of a finite set of states, an input tape, a reading head, and a finite state control. The input tape is divided into squares, each of which contains one symbol. On the tape, in a left-to-right way, $M_1$ works by making a sequence of moves directed by its finite control. During each of these moves, $M_1$ changes its current state and reads the current input symbol $a$, which occurs under the head, and shifts the head one square to the right. Observe that since $M_1$ always shifts its head to the right, reading this particular occurrence of $a$ can be also considered as its erasure. Indeed, once this occurrence of $a$ is read, it is, in effect, gone as well because $M_1$ can never re-read this occurrence during the rest of the move sequence. $M_1$ has one state defined as the start state $s$ and some states designated $M_1$ as final states. With an input string $w$ on its tape, $M_1$ starts working on $w$ from $s$ with the leftmost symbol of $w$ being the current input symbol. If $M_1$ can erase $w$ by a sequence of moves sketched above and, in addition, enter a final state, $M_1$ accepts $w$; otherwise, $M_1$ rejects $w$.

The notion of a two-way finite automaton $M_2$ resembles $M_1$ very much. However, as opposed to the strictly left-to-right behavior of $M_1$, $M_2$ can freely move its head either left or right or remain stationary on the tape. Consequently, the same symbol can be re-read over and over again, so $M_2$ never erases any input symbol on the tape. Just like $M_1$, $M_2$ starts working on $w$ from the start state with the leftmost symbol under its head. If it can make a sequence of moves such that it shifts its head off the right end of the tape and enters a final state, $M_2$ accepts $w$; otherwise, $M_2$ rejects $w$.

Based upon a combination of $M_1$ and $M_2$, we now sketch the new notion of a finite automaton, $M_3$, referred to as an \emph{input-erasing two-way finite automaton}. In essence, $M_3$ works like $M_2$ except that it erases the input symbols just like $M_1$ does. Indeed, once an occurrence of an input symbol is read on the tape, it is erased from the tape (mathematically speaking, this occurrence of the input symbol is changed to the empty word), so $M_3$ can never re-read it again later during its computation. $M_3$ starts working on an input word $w$ from the start state with its input head placed over any symbol occurring in $w$. If $M_3$ can completely erase $w$ by a sequence of left, right, or stationary moves and, in addition, enter a final state, $M_3$ accepts $w$; otherwise, $M_3$ rejects $w$.

As its fundamental result, this paper demonstrates that input-erasing two-way finite automata are stronger than one-way or two-way finite automata, which both characterize the regular language family. Indeed, the input-erasing two-way versions define the linear language family, which properly contains the regular language family. In addition, the paper discusses two kinds of restrictions placed upon input-erasing two-way finite automata and the way they work. The first kind concerns their computation. More precisely, it restricts the performance of left and right moves in a variety of evenly alternating ways and investigates how these restrictions affect their computational power. The other kind explores input-related restrictions. That is, it studies the power of these automata working under the assumption that their input strings or their parts belong to a language family, such as the regular language family.

The present paper is organized as follows. Section \ref{sec:preliminaries} recalls all the terminology needed in this paper. Section \ref{sec:definitions} defines the new type of two-way finite automata. Section \ref{sec:main_result} presents the fundamental results achieved in this paper. Section \ref{sec:computational_restrictions} investigates a variety of evenly alternating restrictions placed upon the way these automata work. Section \ref{sec:input_related_restrictions} explores the various input-related restrictions of these automata. Section \ref{sec:conclusion} closes all the present study by pointing out important open problem areas.
\section{Preliminaries}\label{sec:preliminaries}

For any finite set of nonnegative integers $X$, $\max(X)$ denotes its \emph{maximum}. For any integer $n$, $\abs(n)$ denotes its \emph{absolute value}. For a set $X$, $\card(X)$ denotes the \emph{cardinality} of $X$. Let $X$ and $Y$ be sets; we call $X$ and $Y$ to be \emph{incomparable} if $X \nsubseteq Y$, $Y \nsubseteq X$, and $X \cap Y \neq \emptyset$.

This paper assumes that the reader is familiar with the theory of automata and formal languages (see \cite{AUTOMATA_AND_LANGUAGES_MEDUNA_2000, THEORY_OF_COMPUTATION_WOOD_1987}). For an alphabet $V$, $V^*$ represents the free monoid generated by $V$ under the operation of concatenation; the unit of $V^*$ is denoted by $\epsilon$. Members of $V^*$ are \emph{strings}, and any $L \subseteq V^*$ is a~\emph{formal language}. If $\card(L) = 1$, $L$ is \emph{singular}. Set $V^+ = V^* \setminus \{\epsilon\}$; algebraically, $V^+$ is thus the free semigroup generated by $V$ under the operation of concatenation. For $x \in V^*$, $\strlen{x}$ denotes the length of $x$.

A \emph{finite automaton} (\emph{\fa{}} for short) is a quintuple $M = (Q, \Sigma, R, s, F)$, where $Q$ and $\Sigma$ are two nonempty finite disjoint sets, $R \subseteq Q(\Sigma \cup \{\epsilon\}) \times Q$, $s \in Q$, $F \subseteq Q$. Over $Q\Sigma^*$, we define a binary relation $\Ra$ as follows: for all $(\alpha, q) \in R$ and $u \in \Sigma^*$, $\alpha u \Ra qu$. Extend $\Ra$ to $\Ra^i$ ($i \ge 0$), $\Ra^+$, and $\Ra^*$ in the usual way. Let $L(M) = \{x \mid x \in \Sigma^*, sx \Ra^* f, f \in F\}$. $M$ is \emph{$\epsilon$-free} if $(py, q) \in R$ implies $\strlen{y} = 1$, where $p, q \in Q$, $y \in \Sigma \cup \{\epsilon\}$. In $M$, $Q$, $\Sigma$, $R$, $s$, and $F$ are referred to as the \emph{state} set, the \emph{input alphabet}, the set of \emph{rules}, the \emph{start state}, and the set of \emph{final states}, respectively.

A \emph{two-way finite automaton} (\emph{\twfa{}}) is a quintuple $M = (Q, \Sigma, R, s, F)$, where $Q$, $\Sigma$, $s$, and $F$ has the same meaning as in an \fa{}, and $R \subseteq Q\Sigma \times Q\{\lsh, \rsh\}$, where $\lsh$ and $\rsh$ are two special symbols, $\{\lsh, \rsh\} \cap (Q \cup \Sigma) = \emptyset$. Over $\Sigma^*Q\Sigma^*$, we define a binary relation $\Ra$ as follows:
\begin{enumerate*}[label=(\roman*)]
    \item for all $(pa, q\rsh) \in R$, $upav \Ra uaqv$ and
    \item for all $(pa, q\lsh) \in R$, $ubpav \Ra uqbav$,
\end{enumerate*}
where $a, b \in \Sigma$, $u, v \in \Sigma^*$, $p, q \in Q$. Extend $\Ra$ to $\Ra^i$ ($i \ge 0$), $\Ra^+$, and $\Ra^*$ in the usual way. Let $L(M) = \{x \mid x \in \Sigma^*, sx \Ra^* xf, f \in F\}$.

A \emph{linear grammar} (\emph{\ling{}}) is a quadruple $G = (N, T, P, S)$, where $N$ and $\Sigma$ are two disjoint alphabets, $P \in N \times \Sigma^*(N \cup \{\epsilon\})\Sigma^*$, $S \in N$. Over $T^*(N \cup \{\epsilon\})T^*$, we define a binary relation $\Ra$ as follows: for all $(A, x) \in P$ and $u, v \in T^*$, $uAv \Ra uxv$. Extend $\Ra$ to $\Ra^i$ ($i \ge 0$), $\Ra^+$, and $\Ra^*$ in the usual way. Let $L(M) = \{x \mid x \in T^*, S \Ra^* x\}$. $G$ is an \emph{even linear grammar} (\emph{\eling{}}) if $(A, xBy) \in P$ implies $\strlen{x} = \strlen{y}$, where $x, y \in T^*$, $A, B \in N$. In $G$, $N$ and $T$, are referred to as the alphabets of \emph{nonterminals} and \emph{terminals}, respectively; $P$ is the set of \emph{rules}, and $S$ is the \emph{start nonterminal} of $G$.

In any of the automata or grammars defined above, any $(u, v)$ from $R$ or $P$ is written as $u \to v$ in what follows.

For all $X \in \{\text{\fa{}}, \text{\twfa{}}, \text{\ling{}}, \text{\eling{}}\}$, set $\lfam{}{X} = \{L(Y) \mid Y \text{ is an X}\}$. Let $\lfam{}{\singlfam}$, $\lfam{}{\finlfam}$, $\lfam{}{\reglfam}$ denote the families of singular, finite, and regular languages, respectively. Recall that $\lfam{}{\singlfam} \subset \lfam{}{\finlfam} \subset \lfam{}{\reglfam} = \lfam{}{\fa{}} = \lfam{}{\twfa{}} \subset \lfam{}{\eling{}} \subset \lfam{}{\ling{}}$, where $\lfam{}{\reglfam} \subset \lfam{}{\eling{}}$ is established in \cite{ON_A_FAMILY_OF_LIN_LANGS_AMAR_1964}. Let $\lfam{}{\evenlfam{}}$ and $\lfam{}{\oddlfam{}}$ denote the families of languages consisting of even-length and odd-length strings, respectively.

\section{Definitions}\label{sec:definitions}

In this section, we formally define the general and simple versions of input-erasing two-way finite automata. During every move, the former can read a string, which may consist of several symbols, while the latter always reads no more than one input symbol, just like the classical finite automata do. We also introduce their $\epsilon$-free versions, which cannot perform moves without reading any symbol.

\begin{definition}\label{definition:ietwgfa}
    An \emph{input-erasing two-way general finite automaton} (\emph{\ietwgfa{}}) is a quintuple $$M = (Q, \Sigma, R, s, F),$$ where $Q$ and $\Sigma$ are two nonempty finite disjoint sets, $R \subseteq (Q\Sigma^* \cup \Sigma^*Q) \times Q$, $s \in Q$, $F \subseteq Q$. Let $K = \Sigma^*Q\Sigma^*$. Over $K$, we define a binary relation $\Ra$ as follows: for all $(\alpha, q) \in R$ and $u, v \in \Sigma^*$, $u\alpha v \Ra uqv$. Extend $\Ra$ to $\Ra^i$ ($i \ge 0$), $\Ra^+$, and $\Ra^*$ in the usual way. Let $L(M) = \left\{xy \mid x, y \in \Sigma^*, xsy \Ra^* f, f \in F\right\}$.
\end{definition}

In $M$, $Q$, $\Sigma$, $R$, $s$, and  $F$ are referred to in the same way as in an \fa{}. $K$ is the set of all configurations, and $\Ra$ is the \emph{move} relation. $L(M)$ is the language of $M$.

In what follows, instead of $(\alpha, q) \in R$, we write $\alpha \to q \in R$.
    
Let $r \in R$ be a rule of the form $\alpha \to q$, then $\lhs(r)$ and $\rhs(r)$ denote $\alpha$, called the \emph{left-hand side} of $r$, and $q$, called the \emph{right-hand side} of $r$, respectively.

\begin{definition}
    Let $M = (Q, \Sigma, R, s, F)$  be an \ietwgfa{}, and let $r \in R$. If $\strlen{\lhs(r)} \le 2$, $r$ is \emph{simple}. If $R$ contains only simple rules, $M$ is said to be an \emph{input-erasing two-way simple finite automaton} (\emph{\ietwsfa{}} for short). If $\strlen{\lhs(r)} = 1$, $r$ is an \emph{$\epsilon$-rule} (therefore, every $\epsilon$-rule is simple). If $R$ contains no $\epsilon$-rule, $M$ is \emph{$\epsilon$-free}.
\end{definition}

For both $X \in \{\text{\ietwgfa{}}, \text{\ietwsfa{}}\}$, set $\lfam{\epsilon}{X} = \{L(M) \mid M \text{ is an $X$}\}$ and $\lfam{}{X} = \{L(M) \mid M \text{ is an $\epsilon$-free $X$}\}$.

\section{Main Result}\label{sec:main_result}

The present section demonstrates that \ietwgfa{}s and \ling{}s are equally powerful because they both define the linear language family. Thus, \ietwgfa{}s are stronger than \fa{}s, which characterize the regular language family, properly included in the linear language family.

\begin{lemma}\label{lemma:ietwgfa_in_lg}
    For every \ietwgfa{} $M$, there is an \ling{} $G$ such that $L(G) = L(M)$.
\end{lemma}

\begin{proof}
    Let $M = (Q, \Sigma, R, s, F)$ be an \ietwgfa{} $M$. From $M$, we next construct a \ling{} $G = (N, T, P, S)$ such that $L(G) = L(M)$. Introduce a new symbol $S$---the start nonterminal symbol of $G$. Assume that $S \notin Q$. Set $N = Q \cup \{S\}$ and $T = \Sigma$. Initially, set $P = \{s \to \epsilon\}$. Next, extend $P$ in the following manner:
    \begin{enumerate}
        \item for all $f \in F$, add $S \to f$ to $P$;\label{lemma:ietwgfa_in_lg_step_1}
        \item if $xq \to p \in R$, where $p, q \in Q$ and $x \in \Sigma^*$, add $p \to xq$ to $P$;\label{lemma:ietwgfa_in_lg_step_2}
        \item if $qx \to p \in R$, where $p, q \in Q$ and $x \in \Sigma^*$, add $p \to qx$ to $P$;\label{lemma:ietwgfa_in_lg_step_3}
    \end{enumerate}

    \emph{Basic idea.} $G$ simulates the computation of $M$ in reverse. It starts from the generation of a final state (see step \iref{lemma:ietwgfa_in_lg_step_1}). After this initial derivation step, $G$ simulates every left move made by $M$ according to a rule of the form $xq \to p$, where $p, q \in Q$ and $x \in \Sigma^*$, by using a rule of the form $p \to xq$ (see step \iref{lemma:ietwgfa_in_lg_step_2}). The right moves are simulated analogously (see step \iref{lemma:ietwgfa_in_lg_step_3}). This simulation process is completed by using $s \to \epsilon$, thus erasing the start state $s$ in order to get a string of terminal symbols in~$G$.  

    To establish $L(G) = L(M)$ formally, we next establish the following equivalence.
    
    For all $u, v \in \Sigma^*$ and $p, q \in Q$,
    \begin{equation}\label{claim:ietwgfa_in_lg}
        q \Ra^* upv \text{ in } G \iff upv \Ra^* q \text{ in } M.
    \end{equation}
    
    First, we establish the \emph{only-if} part of this equivalence. That is, by induction on the number of derivation steps $i \ge 0$, we prove that $q \Ra^i upv$ in $G \implies upv \Ra^* q$ in $M$. Let $i = 0$, so $q \Ra^0 upv$ in $G$. Then, $q = p$ and $uv = \epsilon$. Since $q \Ra^0 q$ in $M$, the basis holds true. Assume that the implication holds for all derivations consisting of no more than $j$ steps, for some $j \in \mathbb{N}_0$. Consider any derivation of the form $q \Ra^{j + 1} upv$ in $G$. Let this derivation start with the application of a rule of the form $$q \to xo$$ from $P$, where $o \in Q$ and $x \in \Sigma^*$. Recall that $Q = N \setminus \{S\}$, and observe that $S$ cannot occur on the right-hand side of any rule. Thus, we can express $q \Ra^{j + 1} upv$ as $$q \Ra xo \Ra^j xu'pv$$ in $G$, where $xu' = u$. Then, by the induction hypothesis, $u'pv \Ra^* o$ in $M$. Step \eqref{lemma:ietwgfa_in_lg_step_2} described above constructs $q \to xo \in P$ from $xo \to q \in R$, so $$xu'pv \Ra^* xo \Ra q$$ in $M$. Because $xu' = u$, $upv \Ra^* q$ in $M$. In the case that the derivation $q \Ra^{j + 1} upv$ in $G$ starts with the application of a rule of the form $q \to ox$ from $P$, where $o \in Q$ and $x \in \Sigma^*$, proceed by analogy. Thus, the induction step is completed.

    Next, we establish the \emph{if} part of equivalence \eqref{claim:ietwgfa_in_lg}, so we show that $upv \Ra^i q$ in $M \implies q \Ra^* upv$ in $G$ by induction on the number of moves $i \ge 0$. For $i = 0$, $upv \Ra^0 q$ occurs in $M$ only for $p = q$ and $uv = \epsilon$. Then, since $q \Ra^0 q$ in $G$, the basis holds true. Assume that the implication holds for all computations consisting of no more than $j$ moves, for some $j \in \mathbb{N}_0$. Let $upv \Ra^{j + 1} q$ in $M$, and let this computation end with the application of a rule of the form $$xo \to q$$ from $R$, where $o \in Q$ and $x \in \Sigma^*$. Now, we express $upv \Ra^{j + 1} q$ as $$xu'pv \Ra^j xo \Ra q$$ in $M$, where $xu' = u$. By the induction hypothesis, $o \Ra^* u'pv$ in $G$. From $xo \to q \in R$, step \eqref{lemma:ietwgfa_in_lg_step_2} above constructs $q \to xo \in P$. Thus, $G$ makes $$q \Ra xo \Ra^* xu'pv$$ with $u = xu'$. If the computation $upv \Ra^{j + 1} q$ in $M$ ends with the application of a rule of the form $ox \to q$ from $R$, where $o \in Q$ and $x \in \Sigma^*$, proceed analogously. Thus, the induction step is completed, and equivalence \eqref{claim:ietwgfa_in_lg} holds.

    Considering equivalence \eqref{claim:ietwgfa_in_lg} for $p = s$, for all $u, v \in \Sigma^*$ and $q \in Q$, $q \Ra^* usv \text{ in } G \iff usv \Ra^* q \text{ in } M$. As follows from the presented construction technique, $s \to \epsilon \in P$, and $G$ starts every derivation by applying a rule of the form $S \to f$, where $f \in F$. Consequently, $S \Ra f \Ra^* usv \Ra uv \text{ in } G \iff usv \Ra^* f \text{ in } M$, so $L(G) = L(M)$. Thus, Lemma \ref{lemma:ietwgfa_in_lg} holds.
\end{proof}

\begin{lemma}\label{lemma:lg_in_ietwgfa}
    For every \ling{} $G$, there is an \ietwgfa{} $M$ such that $L(M) = L(G)$.
\end{lemma}

\begin{proof}
    Let $G = (N, T, P, S)$ be a \ling{}. From $G$, we next construct an \ietwgfa{} $M = (Q, \Sigma, R, s , F)$ such that $L(M) = L(G)$. Introduce a new symbol $s$---the start state of $M$. Set $Q' = \{\state{A \to xBy} \mid A \to xBy \in P, A, B \in N, x, y \in T^*\}$. Assume that $Q' \cap N \cap \{s\} = \emptyset$. Set $Q = Q' \cup N \cup \{s\}$, $\Sigma = T$, and $F = \{S\}$. $R$ is constructed as follows:
    \begin{enumerate}
        \item if $A \to x \in P$, where $A \in N$ and $x \in T^*$, add $sx \to A$ to $R$;\label{lemma:lg_in_ietwgfa_step_1}
        \item if $A \to xBy \in P$, where $A, B \in N$ and $x, y \in T^*$, add $xB \to \state{A \to xBy}$, $\state{A \to xBy}y \to A$ to $R$.\label{lemma:lg_in_ietwgfa_step_2}
    \end{enumerate}

    \emph{Basic idea.} $M$ simulates the derivation of $G$ in reverse. It starts by reading a final terminal sequence generated by $G$ (see step \iref{lemma:lg_in_ietwgfa_step_1}). After this initial derivation step, $M$ simulates every derivation made by $G$ according to a rule of the form $A \to xBy$, where $A, B \in N$ and $x, y \in T^*$, by using two consecutive rules of the form $xB \to \state{A \to xBy}$ and $\state{A \to xBy}y \to A$, where $\state{A \to xBy}$ is a newly introduced state with the rule record to which it relates (see step \iref{lemma:lg_in_ietwgfa_step_2}). The entire simulation process is completed by reaching the state $S$ representing the initial nonterminal symbol of $G$ when the input tape of $M$ should be empty.

    In order to demonstrate $L(G) = L(M)$ rigorously, we establish the following equivalence.
    
    For all $u, v \in T^*$ and $A, B \in N$,
    \begin{equation}\label{claim:lg_in_ietwgfa}
        uBv \Ra^* A \text{ in } M \iff A \Ra^* uBv \text{ in } G.
    \end{equation}

    By induction on the number of moves $i \ge 0$, we first prove that $uBv \Ra^i A \text{ in } M \implies A \Ra^* uBv \text{ in } G$. Let $i = 0$, so $uBv \Ra^0 A \text{ in } M$. Then, $A = B$ and $uv = \epsilon$. Clearly, $A \Ra^0 A$ in $G$. For $i = 1$, $uBv \Ra A$ never occurs in $M$ for any $u, v \in T^*$ since, by the construction technique described above, $M$ does not have any rule of the form $xB \to A$ or $By \to A$ for any $x, y \in T^*$. Recall that $A, B \in (Q' \setminus Q) \setminus \{s\} = N$, so no rules added by step \eqref{lemma:lg_in_ietwgfa_step_1} can be applied here. Hence, the basis holds true. Assume that for all $i \le j$, it holds that $uBv \Ra^i A \text{ in } M \implies A \Ra^* uBv \text{ in } G$, for some $j \in \mathbb{N}_0$. Consider any computation of the form $uBv \Ra^{j + 2} A$ in $M$. Let this computation start with the application of two consecutive rules of the form $$xB \to \state{C \to xBy} \text{ and } \state{C \to xBy}y \to C$$ from $R$, where $C \in N$, $x, y \in T^*$, and $C \to xBy \in P$. Thus, we can express $uBv \Ra^{j + 2} A$ as $$u'xByv' \Ra u'\state{C \to xBy}yv' \Ra u'Cv' \Ra^j A$$ in $M$, where $u'x = u$ and $yv' = v$. Clearly, by the induction hypothesis, $A \Ra^* u'Cv'$ in $G$. Step \eqref{lemma:lg_in_ietwgfa_step_2} constructs $xB \to \state{C \to xBy}, \state{C \to xBy}y \to C \in R$ from $C \to xBy \in P$, so $$A \Ra^* u'Cv' \Ra u'xByv'$$ in $G$. Because $u'x = u$ and $yv' = v$, $A \Ra^* uBv$ in $G$, and the induction step is completed.

    Next, by induction on the number of derivation steps $i \ge 0$, we prove that $A \Ra^i uBv$ in $G \implies uBv \Ra^* A$ in $M$. For $i = 0$, $A \Ra^0 uBv$ occurs in $G$ only for $A = B$ and $uv = \epsilon$. Then, the basis holds true since $A \Ra^0 A$ in $M$. Assume that for all $i \le j$, it holds that $A \Ra^i uBv$ in $G \implies uBv \Ra^* A$ in $M$, for some $j \in \mathbb{N}_0$. Let $A \Ra^{j + 1} uBv$ in $G$, and let this derivation end with the application of a rule of the form $$C \to xBy$$ from $P$, where $C \in N$ and $x, y \in T^*$. Now, we express $A \Ra^{j + 1} uBv$ as $$A \Ra^j u'Cv' \Ra u'xByv'$$ in $G$, where $u'x = u$ and $yv' = v$. Hence, by the induction hypothesis, $u'Cv' \Ra^* A$ in $M$. From $C \to xBy \in P$, step \eqref{lemma:lg_in_ietwgfa_step_2} constructs $xB \to \state{C \to xBy}, \state{C \to xBy}y \to C \in R$. Thus, $M$ makes $$u'xByv' \Ra u'\state{C \to xBy}yv' \Ra u'Cv' \Ra^* A$$ with $u'x = u$ and $yv' = v$, and the induction step is completed. Hence, equivalence \eqref{claim:lg_in_ietwgfa} holds.

    Considering equivalence \eqref{claim:lg_in_ietwgfa} for $A = S$, for all $u, v \in T^*$ and $B \in N$, $uBv \Ra^* S \text{ in } M \iff S \Ra^* uBv \text{ in } G$. As follows from the above construction technique, $S \in F$, and $M$ starts each computation by applying a rule of the form $sz \to C$, where $C \in N$, $z \in T^*$, constructed from $C \to z \in P$. Consequently, $uszv \Ra uCv \Ra^* S \text{ in } M \iff S \Ra^* uCv \Ra uzv \text{ in } G$, so $L(M) = L(G)$. Hence, Lemma \ref{lemma:lg_in_ietwgfa} holds.
\end{proof}

\begin{theorem}
    $\lfam{\epsilon}{\ietwgfa{}} = \lfam{\epsilon}{\ling{}}$.
\end{theorem}

\begin{proof}
    As $\lfam{\epsilon}{\ietwgfa} \subseteq \lfam{\epsilon}{\ling}$ follows from Lemma \ref{lemma:ietwgfa_in_lg} and $\lfam{\epsilon}{\ling} \subseteq \lfam{\epsilon}{\ietwgfa}$ from Lemma \ref{lemma:lg_in_ietwgfa}, clearly, the identity $\lfam{\epsilon}{\ietwgfa} = \lfam{\epsilon}{\ling}$ holds.
\end{proof}

\begin{theorem}\label{theorem:ietwgfa_=_ietwsfa_=_efree-ietwsfa}
    $\lfam{\epsilon}{\ietwgfa{}} = \lfam{\epsilon}{\ietwsfa{}} = \lfam{}{\ietwsfa{}}$.
\end{theorem}

\begin{proof}
    As every \ietwsfa{} is a special case of an \ietwgfa{}, we have $\lfam{\epsilon}{\ietwsfa{}} \subseteq \lfam{\epsilon}{\ietwgfa{}}$. To prove $\lfam{\epsilon}{\ietwgfa{}} \subseteq \lfam{\epsilon}{\ietwsfa{}}$, consider any \ietwgfa{} $M$. From $M$, we construct an equivalent \ietwsfa{} $M'$ as follows based upon the following idea. Let $M$ read an $n$-symbol string, $a_1\dots a_n$, to the right during a single move. $M'$ simulates this move as follows:
    \begin{enumerate}
        \item $M'$ records $a_1\dots a_n$ into its current state,
        \item $M'$ makes $n$ subsequent right moves during which it reads $a_1\dots a_n$ symbol by symbol, proceeding from $a_1$ towards $a_n$.
    \end{enumerate}
    The left moves in $M$ are simulated by $M'$ analogously. The details are left to the reader. Thus, $\lfam{\epsilon}{\ietwgfa{}} \subseteq \lfam{\epsilon}{\ietwsfa{}}$ holds, so $\lfam{\epsilon}{\ietwgfa{}} = \lfam{\epsilon}{\ietwsfa{}}$.

    From the definitions, $\lfam{}{\ietwsfa{}} \subseteq \lfam{\epsilon}{\ietwsfa{}}$. The opposite inclusion can be established straightforwardly using the standard technique (see, for instance, Section 3.2.1 in \cite{FORMAL_LANGUAGES_AND_COMPUTATION_MEDUNA_2014}). Thus, $\lfam{\epsilon}{\ietwsfa{}} = \lfam{}{\ietwsfa{}}$. Consequently, Theorem \ref{theorem:ietwgfa_=_ietwsfa_=_efree-ietwsfa} holds.
\end{proof}

\section{Computational Restrictions}\label{sec:computational_restrictions}

This section introduces a variety of restrictions that require the performance of left and right moves in an alternating way. It investigates how these restrictions affect the computational power of \ietwgfa{}s and \ietwsfa{}s. First, however, some additional terminology is needed.

\begin{definition}\label{def:computational_restrictions}
    Let $M = (Q, \Sigma, R, s, F)$ be an \ietwgfa{}, and let $r \in R$. If $\lhs(r) = xq$, $x \in \Sigma^*$, $q \in Q$, $r$ is \emph{left}. Analogously, if $\lhs(r) = qx$, $x \in \Sigma^*$, $q \in Q$, $r$ is \emph{right}.
    
    Let $K$ be the set of all configurations over $M$. Let $\alpha \Ra^* \beta$ in $M$, $\alpha, \beta \in K$. If, in $\alpha \Ra^* \beta$, every sequence of two consecutive moves satisfies the condition that the first of these two moves reads symbols in one direction while the second move reads symbols in the other direction; more precisely, if for every two consecutive moves, $i$ and $j$, in $\alpha \Ra^* \beta$, $i$ is left if and only if $j$ is right, then $\alpha \Ra^* \beta$ is \emph{alternating}, symbolically written as $\alpha \Ra^*_\alt{} \beta$.

    Let $\alpha \Ra^*_\alt \beta$ in $M$ consist of $n$ moves, for some even $n \ge 0$, where $\alpha, \beta \in K$.
    \begin{enumerate}
        \item If, in $\alpha \Ra^*_\alt \beta$, for an odd $i$, $0 \le i \le n$, both the $i$th and the $(i + 1)$th moves read the same number of input symbols, then $\alpha \Ra^*_\alt \beta$ is an \emph{even computation}, symbolically written as $\alpha \Ra^*_\even{} \beta$.
        \item If $\gamma \in K$, $\gamma \Ra \alpha$ in $M$, and $\alpha \Ra^*_\even{} \beta$ in $M$, then $\gamma \Ra \alpha \Ra^*_\even{} \beta$ is an \emph{initialized even computation}, symbolically written as $\gamma \Ra^*_\ieven \beta$.
    \end{enumerate}
    Let $L(M)_\alt{} = \{usv \mid u, v \in \Sigma^*, usv \Ra^*_\alt{} f, f \in F\}$, $L(M)_\even{} = \{usv \mid u, v \in \Sigma^*, usv \Ra^*_\even{} f, f \in F\}$, and $L(M)_\ieven{} = \{usv \mid u, v \in \Sigma^*, usv \Ra^*_\ieven{} f, f \in F\}$.
\end{definition}

For all $X \in \{\text{\ietwgfa{}}, \text{\ietwsfa{}}\}$ and $y \in \{\alt{}, \even{}, \ieven{}\}$, set $\lfam{\epsilon}{X}_y = \{L(M)_y \mid M \text{ is an } X\}$. Analogously, define $\lfam{}{X}_y$ in terms of $\epsilon$-free versions of the corresponding automata.

\begin{theorem}\label{theorem:e-free_alt_ietwsfa_in_alt_ietwsfa}
    $\lfam{}{\ietwsfa{}}_\alt{} \subset \lfam{\epsilon}{\ietwsfa{}}_\alt{}$.
\end{theorem}

\begin{proof}
    \emph{Basic Idea}. Clearly, $\lfam{}{\ietwsfa{}}_\alt{} \subseteq \lfam{\epsilon}{\ietwsfa{}}_\alt{}$. To demonstrate $\lfam{}{\ietwsfa{}}_\alt{} \subset \lfam{\epsilon}{\ietwsfa{}}_\alt{}$, consider $L = \left\{a^nb^nc^m \mid n, m \ge 0\right\}$. Clearly, $L \in \lfam{\epsilon}{\ietwsfa{}}_\alt{}$. Next, we sketch how to prove $L \notin \lfam{}{\ietwsfa{}}_\alt{}$ by the contradiction. Assume that there exists an $\epsilon$-free \ietwsfa{} $M$ such that $L(M)_\alt{} = L$. Take any $a^ib^ic^j$ for some $i, j \ge 0$. $M$ has to start its successful computation in between $a$s and $b$s in order to verify the same number of occurrences of those symbols. After this verification, $M$ has to erase the remaining $j$ $c$s to the right. However, this erasure cannot be performed by $M$, working under alternating computation. Thus, $L \in \lfam{\epsilon}{\ietwsfa{}}_\alt{} \setminus \lfam{}{\ietwsfa{}}_\alt{}$, so Theorem \ref{theorem:e-free_alt_ietwsfa_in_alt_ietwsfa} holds.
\end{proof}

\begin{theorem}\label{theorem:ie2gfa_even_in_e-free_ie2sfa_even}
    For every \ietwgfa{} $M$, there exists an $\epsilon$-free \ietwsfa{} $M'$ such that $L(M') = L(M')_\even{} = L(M)_\even{}$.
\end{theorem}

\begin{proof}
     Let $M = (Q, \Sigma, R, s, F)$ be an \ietwgfa{}. From $M$, we next construct an $\epsilon$-free \ietwsfa{} $M' = (Q', \Sigma, R', s', F')$ such that $L(M') = L(M')_\even{} = L(M)_\even{}$. Introduce a new symbol $s'$---the start state of $M'$. Let $k = \max\{\strlen{\lhs(r)} - 1 \mid r \in R\}$. Set $\hat{Q} = \{\state{xqy\stl{}} \mid q \in Q, x, y \in \Sigma^*, \strlen{x} + \strlen{y} \le 2k - 1, 0 \le \strlen{y} - \strlen{x} \le 1\} \cup \{\state{xqy\str{}} \mid q \in Q, x, y \in \Sigma^*, \strlen{x} + \strlen{y} \le 2k - 1, 0 \le \strlen{x} - \strlen{y} \le 1\}$. Assume that $s' \notin \hat{Q}$. Set $Q' = \hat{Q} \cup \{s'\}$. Initially, set $R' = \emptyset$ and $F' = \{\state{f\stl{}}, \state{f\str{}} \mid f \in F\}$. If $s \in F$, add $s'$ to $F'$. Extend $R'$ by performing \iref{theoremproof:even_ietwgfa_to_efree_ietwsfa_step_first} through \iref{theoremproof:even_ietwgfa_to_efree_ietwsfa_step_4}, given next, until no more rules can be added to $R'$.
    \begin{enumerate}
        \item If $a_1\dots a_np \to q, qa_{n + 1}\dots a_{2n} \to o \in R$, $a_i \in \Sigma$, $1 \le i \le 2n$, $p, q, o \in Q$, for some $n \ge 1$, add
        \begin{align*}
            a_n\state{p\stl{}} &\to \state{a_1\dots a_{n - 1}oa_{n + 1}\dots a_{2n}\stl{}}, \\
            \state{a_1\dots a_{n - 1}oa_{n + 1}\dots a_{2n}\stl{}}a_{n + 1} &\to \state{a_1\dots a_{n - 1}oa_{n + 2}\dots a_{2n}\stl{}}, \\
            a_{n - 1}\state{a_1\dots a_{n - 1}oa_{n + 2}\dots a_{2n}\stl{}} &\to \state{a_1\dots a_{n - 2}oa_{n + 2}\dots a_{2n}\stl{}}, \\
            &\vdotswithin{\to} \\
            \state{oa_{2n}\stl{}}a_{2n} &\to \state{o\stl{}}
        \end{align*}
        to $R'$; in addition, if $p = s$, include $a_ns' \to \state{a_1\dots a_{n - 1}oa_{n + 1}\dots a_{2n}\stl{}}$ into $R'$.\label{theoremproof:even_ietwgfa_to_efree_ietwsfa_step_first}
        \item If $pa_n\dots a_1 \to q, a_{2n}\dots a_{n + 1}q \to o \in R$, $a_i \in \Sigma$, $1 \le i \le 2n$, $p, q, o \in Q$, for some $n \ge 1$, add
        \begin{align*}
            \state{p\str{}}a_n &\to \state{a_{2n}\dots a_{n + 1}oa_{n - 1}\dots a_1\str{}}, \\
            a_{n + 1}\state{a_{2n}\dots a_{n + 1}oa_{n - 1}\dots a_1\str{}} &\to \state{a_{2n}\dots a_{n + 2}oa_{n - 1}\dots a_1\str{}}, \\
            \state{a_{2n}\dots a_{n + 2}oa_{n - 1}\dots a_1\str{}}a_{n - 1} &\to \state{a_{2n}\dots a_{n + 2}oa_{n - 2}\dots a_1\str{}}, \\
            &\vdotswithin{\to} \\
            a_{2n}\state{a_{2n}o\str{}} &\to \state{o\str{}}
        \end{align*}
        to $R'$; in addition, if $p = s$, add $s'a_n \to \state{a_{2n}\dots a_{n + 1}oa_{n - 1}\dots a_1\str{}}$ to $R'$, too.\label{theoremproof:even_ietwgfa_to_efree_ietwsfa_step_2}
        \item For all $p \to q, q \to o \in R$, $o, p, q \in Q$, and $r' \in R'$ with $\lhs(r') = a\state{o\stl{}}$, $a \in \Sigma$, add $a\state{p\stl{}} \to \rhs(r')$ to $R'$; in addition, if $p = s$, add $as' \to \rhs(r')$ to $R'$, too.\label{theoremproof:even_ietwgfa_to_efree_ietwsfa_step_3}
        \item For all $p \to q, q \to o \in R$, $o, p, q \in Q$, and $r' \in R'$ with $\lhs(r') = \state{o\str{}}a$, $a \in \Sigma$, add $\state{p\str{}}a \to \rhs(r')$ to $R'$; moreover, if $p = s$, also add $s'a \to \rhs(r')$ to $R'$.\label{theoremproof:even_ietwgfa_to_efree_ietwsfa_step_4}
    \end{enumerate}
    Repeat the following extension of $F'$ until no more states can be included in $F'$.   
    \begin{enumerate}[resume]
       \item For all $p \to q, q \to o \in R$, where $o, p, q \in Q$ and $\state{o\stl{}}, \state{o\str{}} \in F'$, add $\state{p\stl{}}$ and $\state{p\str{}}$ to $F'$; in addition, if $p = s$, add $s'$ to $F'$.\label{theoremproof:even_ietwgfa_to_efree_ietwsfa_step_last}
    \end{enumerate}

    \emph{Basic idea.} As is obvious, $M'$ represents an $\epsilon$-free \ietwsfa{}. $M'$ simulates any even computation in $M$ by making sequences of moves, each of which erases a single symbol. To explain step \iref{theoremproof:even_ietwgfa_to_efree_ietwsfa_step_first}, assume that $M$ performs a two-move even computation by rules $a_1\dots a_np \to q, qa_{n + 1}\dots a_{2n} \to o \in R$. Consider the sequence of rules introduced into $R'$ in step \iref{theoremproof:even_ietwgfa_to_efree_ietwsfa_step_first}. Observe that once  $M'$ applies its first rule, $M'$ has to apply all the remaining rules of this sequence in an uninterrupted one-by-one way, and thereby, it simulates the two-step computation in $M$. Notice that the first rule, $a_n\state{p\stl{}} \to \state{a_1\dots a_{n - 1}oa_{n + 1}\dots a_{2n}\stl{}}$, is a left rule. Step \iref{theoremproof:even_ietwgfa_to_efree_ietwsfa_step_2} is analogous to step \iref{theoremproof:even_ietwgfa_to_efree_ietwsfa_step_first}, except that the first rule of the introduced sequence is a right rule. To explain step \iref{theoremproof:even_ietwgfa_to_efree_ietwsfa_step_3}, assume that 
    \begin{enumerate*}[label=(\roman*)]
        \item $M$ performs an even computation according to two $\epsilon$-rules $p \to q, q \to o \in R$ and that
        \item $R'$ contains $r'$ with $\lhs(r') = a\state{o\stl{}}$ ($r'$ is introduced into $R'$ in step \iref{theoremproof:even_ietwgfa_to_efree_ietwsfa_step_first} or \iref{theoremproof:even_ietwgfa_to_efree_ietwsfa_step_3}).
    \end{enumerate*}
    Then, this step introduces $a\state{p\stl{}} \to \rhs(r')$ into $R'$. By using this newly introduced rule, $a\state{p\stl{}} \to \rhs(r')$, $M'$ actually skips over the two-move even computation according to $p \to q$ and $q \to o$ in $M$, after which it enters the state $\rhs(r')$, which occurs as the right-hand side of the first rule of a rule sequence introduced in step \iref{theoremproof:even_ietwgfa_to_efree_ietwsfa_step_first}. Step \iref{theoremproof:even_ietwgfa_to_efree_ietwsfa_step_4} parallels step \iref{theoremproof:even_ietwgfa_to_efree_ietwsfa_step_3}, except that $r'$ is a right rule in step \iref{theoremproof:even_ietwgfa_to_efree_ietwsfa_step_4}, while it is a left rule in step \iref{theoremproof:even_ietwgfa_to_efree_ietwsfa_step_3}.

    Consider $F'$. Assume that an even accepting computation in $M$ ends with an even sequence of moves according to $\epsilon$-rules (including the empty sequence). Observe that at this point, by the extension of $F'$ from step \iref{theoremproof:even_ietwgfa_to_efree_ietwsfa_step_last}, $M'$ accepts, too.

    To establish $L(M')_\even{} = L(M)_\even{}$ formally, we first prove the following two equivalence.
    
    When $M$ is $\epsilon$-free, for all $u, v \in \Sigma^*$ and $p, q \in Q$,
    \begin{equation}\label{claim:ie2gfa_even_in_efree_ie2sfa_even_A}
        u\state{p\stl{}}v \Ra_{\even{}}^* \state{q\stl{}} \text{ in } M' \iff upv \Ra^*_\even{} q \text{ in } M,
    \end{equation}
    where $upv \Ra^*_\even{} q$ starts with a left move (unless it consists of no moves).

    We begin by proving the \emph{only-if} part of equivalence \ref{claim:ie2gfa_even_in_efree_ie2sfa_even_A}. That is, by induction on the number of moves $i \ge 0$, we show that for $\epsilon$-free $M$, $u\state{p\stl{}}v \Ra_{\even{}}^i \state{q\stl{}} \text{ in } M'$ implies $upv \Ra^*_\even{} q \text{ in } M$, where $upv \Ra^*_\even{} q$ starts with a left move (or consists of no moves). Let $i = 0$, so $u\state{p\stl{}}v \Ra_{\even{}}^0 \state{q\stl{}} \text{ in } M'$. Then, $p = q$ and $uv = \epsilon$. Clearly, $q \Ra^0_\even{} q$ in $M$. Let $i = 1$. Then, $u\state{p\stl{}}v \Ra_{\even{}}^1 \state{q\stl{}}$ never occurs in $M'$ because, by Definition \ref{def:computational_restrictions}, each even computation is supposed to have an even number of moves; however, $u\state{p\stl{}}v \Ra_{\even{}}^1 \state{q\stl{}}$ has only one move. Thus, the basis holds true. Assume that the implication holds for all computations consisting of no more than $j$ moves in $M'$, for some $j \in \mathbb{N}_0$. Consider any computation of the form $u\state{p\stl{}}v \Ra_{\even{}}^{j + 2n} \state{q\stl{}}$ in $M'$, for some $n \ge 1$. Let this computation start with the application of $2n$ consecutive rules of the form
    \begin{align*}
        a_n\state{p\stl{}} &\to \state{a_1\dots a_{n - 1}oa_{n + 1}\dots a_{2n}\stl{}}, \\
        \state{a_1\dots a_{n - 1}oa_{n + 1}\dots a_{2n}\stl{}}a_{n + 1} &\to \state{a_1\dots a_{n - 1}oa_{n + 2}\dots a_{2n}\stl{}}, \\
        a_{n - 1}\state{a_1\dots a_{n - 1}oa_{n + 2}\dots a_{2n}\stl{}} &\to \state{a_1\dots a_{n - 2}oa_{n + 2}\dots a_{2n}\stl{}}, \\
        &\vdotswithin{\to} \\
        \state{oa_{2n}\stl{}}a_{2n} &\to \state{o\stl{}}
    \end{align*}
    from $R'$, where $o \in Q$ and $a_k \in \Sigma$ for all $1 \le k \le 2n$. Thus, we can express $u\state{p\stl{}}v \Ra_{\even{}}^{j + 2n} \state{q\stl{}}$ as
    \begin{align*}
        u'a_1\ldots a_n\state{p\stl{}}a_{n + 1}\ldots a_{2n}v' &\Ra u'a_1\ldots a_{n - 1}\state{a_1\dots a_{n - 1}oa_{n + 1}\dots a_{2n}\stl{}}a_{n + 1}\ldots a_{2n}v' \\
        &\Ra u'a_1\ldots a_{n - 1}\state{a_1\dots a_{n - 1}oa_{n + 2}\dots a_{2n}\stl{}}a_{n + 2}\ldots a_{2n}v' \\
        &\Ra u'a_1\ldots a_{n - 2}\state{a_1\dots a_{n - 2}oa_{n + 2}\dots a_{2n}\stl{}}a_{n + 2}\ldots a_{2n}v' \\
        &\vdotswithin{\Ra} \\
        &\Ra u'\state{oa_{2n}\stl{}}a_{2n}v' \Ra u'\state{o\stl{}}v' \Ra^j_\even{} \state{q\stl{}}
    \end{align*}
    in $M'$, where $u'a_1\ldots a_n = u$ and $a_{n + 1}\ldots a_{2n}v' = v$. According to the induction hypothesis, $u'ov' \Ra^*_\even{} q$ in $M$, and this computation starts with a left move (or consists of no moves at all). Step \iref{theoremproof:even_ietwgfa_to_efree_ietwsfa_step_first} above constructs $a_n\state{p\stl{}} \to \state{a_1\dots a_{n - 1}oa_{n + 1}\dots a_{2n}\stl{}},\allowbreak \ldots, \state{oa_{2n}\stl{}}a_{2n} \to \state{o\stl{}} \in R'$ from $a_1\dots a_np \to t, ta_{n + 1}\dots a_{2n} \to o \in R$, for some $t \in Q$, so $M$ makes $$u'a_1\dots a_npa_{n + 1}\dots a_{2n}v' \Ra u'ta_{n + 1}\dots a_{2n}v' \Ra u'ov' \Ra^*_\even{} q.$$ Taking into account the properties of $u'ov' \Ra^*_\even{} q$, since $u'a_1\ldots a_n = u$, $a_{n + 1}\ldots a_{2n}v' = v$, it follows that $upv \Ra^*_\even{} q$ in $M$. As we can see, $upv \Ra^*_\even{} q$ starts with a~left move, which completes the induction step.

    Next, we prove the \emph{if} part of equivalence \eqref{claim:ie2gfa_even_in_efree_ie2sfa_even_A}. By induction on the number of moves $i \ge 0$, we show that for $\epsilon$-free $M$, $upv \Ra^i_\even{} q$, which starts with a left move (or consists of no moves), in $M \implies u\state{p\stl{}}v \Ra_{\even{}}^* \state{q\stl{}} \text{ in } M'$. Let $i = 0$, so $upv \Ra^0_\even{} q \text{ in } M$. Then, $p = q$ and $uv = \epsilon$. Clearly, $\state{q\stl{}} \Ra_{\even{}}^0 \state{q\stl{}}$ in $M'$. For $i = 1$, $upv \Ra^1_\even{} q$ never occurs in $M$ since, by the definition of even computation (see Definition \ref{def:computational_restrictions}), every $upv \Ra^*_\even{} q$ consists of an even number of moves; however, $upv \Ra^1_\even{} q$ consists of a single move only. Thus, the basis holds true. Assume that the implication holds for all computations consisting of no more than $j$ moves in $M$, for some $j \in \mathbb{N}_0$. Consider any computation of the form $upv \Ra^{j + 2}_\even{} q$ in $M$. Let this computation start with the application of two consecutive rules of the form $$a_1\dots a_np \to t \text{ and } ta_{n + 1}\dots a_{2n} \to o$$ from $R$, where $o, t \in Q$ and $a_k \in \Sigma$ for all $1 \le k \le 2n$, for some $n \ge 1$. Hence, we can express $upv \Ra^{j + 2}_\even{} q$ as $$u'a_1\dots a_npa_{n + 1}\dots a_{2n}v' \Ra u'ta_{n + 1}\dots a_{2n}v' \Ra u'ov' \Ra^j_\even{} q$$ in $M$, where $u'a_1\dots a_n = u$ and $a_{n + 1}\dots a_{2n}v' = v$. According to the induction hypothesis $u'\state{o\stl{}}v' \Ra_{\even{}}^* \state{q\stl{}} \text{ in } M'$. From $a_1\dots a_np \to t, ta_{n + 1}\dots a_{2n} \to o \in R$, step \iref{theoremproof:even_ietwgfa_to_efree_ietwsfa_step_first} constructs
    \begin{align*}
        a_n\state{p\stl{}} &\to \state{a_1\dots a_{n - 1}oa_{n + 1}\dots a_{2n}\stl{}}, \\
        \state{a_1\dots a_{n - 1}oa_{n + 1}\dots a_{2n}\stl{}}a_{n + 1} &\to \state{a_1\dots a_{n - 1}oa_{n + 2}\dots a_{2n}\stl{}}, \\
        a_{n - 1}\state{a_1\dots a_{n - 1}oa_{n + 2}\dots a_{2n}\stl{}} &\to \state{a_1\dots a_{n - 2}oa_{n + 2}\dots a_{2n}\stl{}}, \\
        &\vdotswithin{\to} \\
        \state{oa_{2n}\stl{}}a_{2n} &\to \state{o\stl{}} \in R',
    \end{align*}
    so $M'$ makes
    \begin{align*}
        u'a_1\ldots a_n\state{p\stl{}}a_{n + 1}\ldots a_{2n}v' &\Ra u'a_1\ldots a_{n - 1}\state{a_1\dots a_{n - 1}oa_{n + 1}\dots a_{2n}\stl{}}a_{n + 1}\ldots a_{2n}v' \\
        &\Ra u'a_1\ldots a_{n - 1}\state{a_1\dots a_{n - 1}oa_{n + 2}\dots a_{2n}\stl{}}a_{n + 2}\ldots a_{2n}v' \\
        &\Ra u'a_1\ldots a_{n - 2}\state{a_1\dots a_{n - 2}oa_{n + 2}\dots a_{2n}\stl{}}a_{n + 2}\ldots a_{2n}v' \\
        &\vdotswithin{\Ra} \\
        &\Ra u'\state{oa_{2n}\stl{}}a_{2n}v' \Ra u'\state{o\stl{}}v' \Ra^*_\even{} \state{q\stl{}}.
    \end{align*}
    Notice that by the construction technique of $M'$, $u'\state{o\stl{}}v' \Ra^*_\even{} \state{q\stl{}}$ can never starts with a~right move. This, along with the fact that $u'a_1\ldots a_n = u$ and $a_{n + 1}\ldots a_{2n}v' = v$, implies that $u\state{p\stl{}}v \Ra_{\even{}}^* \state{q\stl{}}$ in $M'$. Thus, the induction step is completed, and equivalence \eqref{claim:ie2gfa_even_in_efree_ie2sfa_even_A} holds.

    When $M$ is $\epsilon$-free, for all $u, v \in \Sigma^*$ and $p, q \in Q$,
    \begin{equation}\label{claim:ie2gfa_even_in_efree_ie2sfa_even_B}
        u\state{p\str{}}v \Ra_{\even{}}^* \state{q\str{}} \text{ in } M' \iff upv \Ra^*_\even{} q \text{ in } M,
    \end{equation}
    where $upv \Ra^*_\even{} q$ starts with a right move (unless it consists of no moves).

    Prove equivalence \eqref{claim:ie2gfa_even_in_efree_ie2sfa_even_B} by analogy with the proof of equivalence \eqref{claim:ie2gfa_even_in_efree_ie2sfa_even_A}.

    Equivalences \eqref{claim:ie2gfa_even_in_efree_ie2sfa_even_A} and \ref{claim:ie2gfa_even_in_efree_ie2sfa_even_B} demonstrate the correctness of steps \iref{theoremproof:even_ietwgfa_to_efree_ietwsfa_step_first} and \iref{theoremproof:even_ietwgfa_to_efree_ietwsfa_step_2} from the above construction technique. However, it does not address the elimination of $\epsilon$-rules of $M$ in steps \iref{theoremproof:even_ietwgfa_to_efree_ietwsfa_step_3}, \iref{theoremproof:even_ietwgfa_to_efree_ietwsfa_step_4}, and \iref{theoremproof:even_ietwgfa_to_efree_ietwsfa_step_last}. For this reason, we next establish equivalences \eqref{claim:ie2gfa_even_in_efree_ie2sfa_even_C}, \eqref{claim:ie2gfa_even_in_efree_ie2sfa_even_D}, and \eqref{claim:ie2gfa_even_in_efree_ie2sfa_even_E}.

    For all $x \in \Sigma^*, y \in \Sigma^+, a \in \Sigma$, and $p, t \in Q$ such that $\strlen{x} + 1 = \strlen{y}$,
    \begin{equation}\label{claim:ie2gfa_even_in_efree_ie2sfa_even_C}
        a\state{p\stl{}} \Ra \state{xty\stl{}} \text{ in } M' \iff \text{there are $o, q \in Q$ such that } xapy \Ra^*_\even{} xaqy \Ra oy \Ra t \text{ in } M.
    \end{equation}

    First, we establish the \emph{only-if} part of equivalence \eqref{claim:ie2gfa_even_in_efree_ie2sfa_even_C}. By induction on $i \ge 0$, which represents the number of iterations of step \iref{theoremproof:even_ietwgfa_to_efree_ietwsfa_step_3}, we show that $a\state{p\stl{}} \Ra \state{xty\stl{}} \text{ in } M'$ implies that there are $o, q \in Q$ such that $xapy \Ra^*_\even{} xaqy \Ra oy \Ra t \text{ in } M$. For $i = 0$, $a\state{p\stl{}} \Ra \state{xty\stl{}} \text{ in } M'$ can only be performed using a rule of the form $a\state{p\stl{}} \to \state{xty\stl{}}$ added to $R'$ in step \iref{theoremproof:even_ietwgfa_to_efree_ietwsfa_step_first}. Then, since step \iref{theoremproof:even_ietwgfa_to_efree_ietwsfa_step_first} constructs $a\state{p\stl{}} \to \state{xty\stl{}} \in R'$ from $xap \to g, gy \to t \in R$, for some $g \in Q$, it follows that $xapy \Ra gy \Ra t$ in $M$. Thus, the basis holds true. Assume that the implication holds for no more than $j$ iterations of step \iref{theoremproof:even_ietwgfa_to_efree_ietwsfa_step_3}, for some $j \in \mathbb{N}_0$. Consider any $a\state{p\stl{}} \Ra \state{xty\stl{}} \text{ in } M'$ performed using a~rule of the form $a\state{p\stl{}} \to \state{xty\stl{}}$ that belongs to $R'$ from the $(j + 1)$th iteration of step \iref{theoremproof:even_ietwgfa_to_efree_ietwsfa_step_3}. From this, it follows that there exist $p \to g, g \to h \in R$, for some $g, h \in Q$, and $a\state{h\stl{}} \to \state{xty\stl{}} \in R'$ that was added to $R'$ during the $j$th iteration of step \iref{theoremproof:even_ietwgfa_to_efree_ietwsfa_step_3}. By the induction hypothesis, there are $o, q \in Q$ such that $xahy \Ra^*_\even{} xaqy \Ra oy \Ra t$ in $M$, so $M$ can make $$xapy \Ra xagy \Ra xahy \Ra^*_\even{} xaqy \Ra oy \Ra t.$$ By the definition of even computation, $xapy \Ra^*_\even{} xaqy \Ra oy \Ra t$ in $M$. Hence, the induction step is completed.

    Next, we establish the \emph{if} part of equivalence \eqref{claim:ie2gfa_even_in_efree_ie2sfa_even_C}. By induction on the number of moves $i \ge 0$, we show that $xapy \Ra^i_\even{} xaqy \Ra oy \Ra t \text{ in } M \implies a\state{p\stl{}} \to \state{xty\stl{}} \in R'$. Let $i = 0$, so $xapy \Ra^0_\even{} xaqy \Ra oy \Ra t \text{ in } M$. Then, $p = q$. Clearly, according to step \iref{theoremproof:even_ietwgfa_to_efree_ietwsfa_step_first}, $a\state{q\stl{}} \to \state{xty\stl{}} \in R'$, so $a\state{q\stl{}} \Ra \state{xty\stl{}} \text{ in } M'$. Let $i = 1$, so $xapy \Ra^1_\even{} xaqy \Ra oy \Ra t \text{ in } M$. This never happens because, by the definition of even computation, every computation of the form $p \Ra^*_\even{} q$ consists of an even number of moves; however, $p \Ra^1_\even{} q$ consists of only one, which is an odd number. Thus, the basis holds true. Assume that the implication holds for all computations of the form $xapy \Ra^{k}_\even{} xaqy \Ra oy \Ra t \text{ in } M$ with $0 \le k \le j$, for some $j \in \mathbb{N}_0$. Consider any computation of the form $xapy \Ra^{j + 2}_\even{} xaqy \Ra oy \Ra t \text{ in } M$, and let it start with the application of two consecutive rules of the form $$p \to g \text{ and } g \to h$$ from $R$, where $g, h \in Q$. Then, we can express $xapy \Ra^{j + 2}_\even{} xaqy \Ra oy \Ra t$ as $$xapy \Ra xagy \Ra xahy \Ra^j_\even{} xaqy \Ra oy \Ra t$$ in $M$. Clearly, by the induction hypothesis, $a\state{h\stl{}} \Ra \state{xty\stl{}}$ in $M'$. Hence, $a\state{h\stl{}} \to \state{xty\stl{}} \in R'$. From $a\state{h\stl{}} \to \state{xty\stl{}} \in R'$ and $p \to g, g \to h \in R$, step \iref{theoremproof:even_ietwgfa_to_efree_ietwsfa_step_3} constructs $a\state{p\stl{}} \to \state{xty\stl{}} \in R'$, so $a\state{p\stl{}} \Ra \state{xty\stl{}}$ in $M'$. Thus, the induction step is completed, and equivalence \eqref{claim:ie2gfa_even_in_efree_ie2sfa_even_C} holds.

    For all $x \in \Sigma^+, y \in \Sigma^*, a \in \Sigma$, and $p, t \in Q$ such that $\strlen{x} = \strlen{y} + 1$,
    \begin{equation}\label{claim:ie2gfa_even_in_efree_ie2sfa_even_D}
        \state{p\str{}}a \to \state{xty\str{}} \in R' \iff \text{there are $o, q \in Q$ such that } xpay \Ra^*_\even{} xqay \Ra xo \Ra t \text{ in } M.
    \end{equation}

    Prove equivalence \ref{claim:ie2gfa_even_in_efree_ie2sfa_even_D} by analogy with the proof of equivalence \eqref{claim:ie2gfa_even_in_efree_ie2sfa_even_C}.

    For all $p \in Q$,
    \begin{equation}\label{claim:ie2gfa_even_in_efree_ie2sfa_even_E}
        \state{p\stl{}}, \state{p\str{}} \in F' \iff \text{there is $f \in F$ such that } p \Ra^*_\even{} f \text{ in } M.
    \end{equation}

    First, we establish the \emph{only-if} part of equivalence \eqref{claim:ie2gfa_even_in_efree_ie2sfa_even_E}. By induction on $i \ge 0$, which represents the number of iterations of step \iref{theoremproof:even_ietwgfa_to_efree_ietwsfa_step_last}, we prove that $\state{q\stl{}}, \state{q\str{}} \in F'$ implies that there is $f \in F$ such that $q \Ra^*_\even{} f \text{ in } M$. For $i = 0$, only $\state{f\stl{}}, \state{f\str{}} \in F'$ for all $f \in F$. Clearly, $f \Ra^0 f$ in $M$, so the basis holds true. Assume that the implication holds for no more than $j$ iterations of step \iref{theoremproof:even_ietwgfa_to_efree_ietwsfa_step_last}, for some $j \in \mathbb{N}_0$. Consider any $\state{p\stl{}}, \state{p\str{}} \in Q'$ that belong to $F'$ since the $(j + 1)$th iteration of step \iref{theoremproof:even_ietwgfa_to_efree_ietwsfa_step_last}. Then, there exist $p \to q, p \to o \in R$, for some $o, q \in Q$, and $\state{o\stl{}}, \state{o\str{}} \in Q'$ that were added to $F'$ during the $j$th iteration of step \iref{theoremproof:even_ietwgfa_to_efree_ietwsfa_step_last}. By the induction hypothesis, there is $f \in F$ such that $o \Ra^*_\even{} f$ in $M$, so $M$ can make $$p \Ra q \Ra o \Ra^*_\even{} f.$$ Hence, by the definition of even computation, $p \Ra^*_\even{} f$, and the induction step is completed.

    Now, we establish the \emph{if} part of equivalence \eqref{claim:ie2gfa_even_in_efree_ie2sfa_even_E}. By induction on the number of moves $i \ge 0$, we show that $p \Ra^i_\even{} f$, where $f \in F$, in $M$ implies $\state{p\stl{}}, \state{p\str{}} \in F'$. Let $i = 0$, so $p \Ra^0_\even{} f$ in $M$. Then, $p = f$. Clearly, $\state{f\stl{}}, \state{f\str{}} \in F'$ as $\state{q\stl{}}, \state{q\str{}} \in F'$ for all $q \in F$. For $i = 1$, $p \Ra^1_\even{} f$ never occurs in $M$ because, by Definition \ref{def:computational_restrictions}, every even computation is supposed to have an even number of moves. However, $p \Ra^1_\even{} f$ has only one move. Therefore, the basis holds true. Assume that the implication holds for all computations consisting of no more than $j$ moves in $M$, for some $j \in \mathbb{N}_0$. Consider any computation of the form $p \Ra^{j + 2}_\even{} f$ in $M$ with $f \in F$. Let this computation start with the application of two consecutive rules of the form $$p \to q \text{ and } q \to o$$ from $R$, where $o, q \in Q$. Thus, we can express $p \Ra^{j + 2}_\even{} f$ as $$p \Ra q \Ra o \Ra^j_\even{} f$$ in $M$. By the induction hypothesis, $\state{o\stl{}}, \state{o\str{}} \in F'$. Since $\state{o\stl{}}, \state{o\str{}} \in F'$ and $p \to q, q \to o \in R$, step \iref{theoremproof:even_ietwgfa_to_efree_ietwsfa_step_last} adds $\state{p\stl{}}$ and $\state{p\str{}}$ to $F'$. Thus, the induction step is completed, and equivalence \eqref{claim:ie2gfa_even_in_efree_ie2sfa_even_E} holds.

    Based on equivalences \eqref{claim:ie2gfa_even_in_efree_ie2sfa_even_A}, \eqref{claim:ie2gfa_even_in_efree_ie2sfa_even_B}, \eqref{claim:ie2gfa_even_in_efree_ie2sfa_even_C}, \eqref{claim:ie2gfa_even_in_efree_ie2sfa_even_D}, and \eqref{claim:ie2gfa_even_in_efree_ie2sfa_even_E} above, we can conclude that for all $u, v \in \Sigma^*$ and $p, q \in Q$, $u\state{p\stl{}}v \Ra^*_\even{} \state{q\stl{}}$ or $u\state{p\str{}}v \Ra^*_\even{} \state{q\str{}}$ in $M'$, where $\state{q\stl{}}, \state{q\str{}} \in F'$, iff there is $f \in F$ such that $upv \Ra^*_\even{} q \Ra^*_\even{} f$ in $M$. Considering this equivalence for $p = s$, $u\state{s\stl{}}v \Ra^*_\even{} \state{q\stl{}}$ or $u\state{s\str{}}v \Ra^*_\even{} \state{q\str{}}$ in $M'$, where $\state{q\stl{}}, \state{q\str{}} \in F'$, iff there is $f \in F$ such that $usv \Ra^*_\even{} q \Ra^*_\even{} f$ in $M$. As follows from the construction technique, $M'$ starts every computation from its initial state $s'$, from which the same moves can be made as from the states $\state{s\stl{}}$ and $\state{s\str{}}$. In other words, $M'$ starts each computation using either a rule of the form $as' \to t'$, for which $a\state{s\stl{}} \to t' \in R'$, or a rule of the form $s'a \to t'$, for which $\state{s\str{}}a \to t' \in R'$, where $a \in \Sigma$ and $t' \in Q'$. Consequently, $us'v \Ra^*_\even{} \state{q\stl{}}$ or $us'v \Ra^*_\even{} \state{q\str{}}$ in $M'$, where $\state{q\stl{}}, \state{q\str{}} \in F'$, iff there is $f \in F$ such that $usv \Ra^*_\even{} q \Ra^*_\even{} f$ in $M$. Hence, $L(M')_\even{} = L(M)_\even{}$.

    Obviously, $L(M')_\even{} \subseteq L(M')$ follows directly from the definition of even computation. The opposite inclusion, $L(M') \subseteq L(M')_\even{}$, follows from the construction technique above. Indeed, for each state of $M'$ except $s'$, according to the rules in $R'$, all incoming moves read symbols in one direction, while all outgoing moves read symbols in the opposite direction. From $s'$, both left and right moves can be made, as no move ever leads to this state. Therefore, Theorem \ref{theorem:ie2gfa_even_in_e-free_ie2sfa_even} holds.
\end{proof}

\begin{theorem}\label{theorem:ie2gfa_even_in_even}
    $\lfam{\epsilon}{\ietwgfa{}}_\even{} \subset \lfam{}{\evenlfam}$.
\end{theorem}

\begin{proof}
    As each even computation consists of an even number of moves, each language in $\lfam{\epsilon}{\ietwgfa{}}_\even{}$ clearly contains even-length strings only. Thus, $\lfam{\epsilon}{\ietwgfa{}}_\even{} \subset \lfam{}{\evenlfam}$.
\end{proof}

\begin{theorem}\label{theorem:even_ietwgfa_incomparable_with_sing_fin_reg}
    $\lfam{\epsilon}{\ietwgfa{}}_\even{}$ is incomparable with any of these language families---$\lfam{}{\singlfam}$, $\lfam{}{\finlfam}$, and $\lfam{}{\reglfam}$.
\end{theorem}

\begin{proof}
    Let $L \in \lfam{\epsilon}{\ietwgfa{}}_\even{}$. By Theorem \ref{theorem:ie2gfa_even_in_even}, $x \in L$ implies that $\strlen{x}$ is even. Thus, any $\{y\} \in \lfam{}{\singlfam{}}$ with $\strlen{y}$ being odd, such as $\{a\}$, is outside of $\lfam{\epsilon}{\ietwgfa{}}_\even{}$. Clearly, $\{aa\} \in \lfam{\epsilon}{\ietwgfa{}}_\even{} \cap \lfam{}{\singlfam{}}$. Notice that $\{a^nb^n \mid n \ge 0\} \in \lfam{\epsilon}{\ietwgfa{}}_\even{} \setminus \lfam{}{\singlfam{}}$. The rest of this proof is left to the reader as it follows the same reasoning.
\end{proof}

\begin{theorem}\label{theorem:ieven_ietwgfa_to_ietwsfa}
    For every \ietwgfa{} $M$, there is an \ietwsfa{} $M' = (Q', \Sigma', R', s', F')$ such that
    \begin{enumerate}[label=(\roman*)]
        \item $r \in R'$ implies $\rhs(r) \neq s'$, and $\strlen{\lhs(r)} = 1$ implies $\lhs(r) = s'$;
        \item $L(M') = L(M')_\ieven{} = L(M)_\ieven{}$.
    \end{enumerate}
\end{theorem}

\begin{proof}
    Let $M = (Q, \Sigma, R, s, F)$ be an \ietwgfa{}. Next, we construct an \ietwsfa{} $M' = (Q', \Sigma, R', s', F')$ satisfying the properties of Theorem \ref{theorem:ieven_ietwgfa_to_ietwsfa}. Let $\hat{M} = (\hat{Q}, \Sigma, \hat{R}, \hat{s}, \hat{F})$ be the $\epsilon$-free \ietwsfa{} constructed from $M$ by the technique described in the proof of Theorem \ref{theorem:ie2gfa_even_in_e-free_ie2sfa_even}. Recall that $L(\hat{M}) = L(\hat{M})_\even{} = L(M)_\even{}$ and $\state{q\stl{}}, \state{q\str{}} \in \hat{Q}$ for all $q \in Q$. Introduce a new symbol $s'$---the start state of $M'$. Set $\bar{Q} = \{\state{xqy} \mid q \in Q, x, y \in \Sigma^*, 1 \le \strlen{x} + \strlen{y} \le k - 1, \abs(\strlen{x} - \strlen{y}) \le 1\}$, where $k = \max\{\strlen{\lhs(r)} \mid r \in R\}$. Assume that $\hat{Q} \cap \bar{Q} \cap \{s'\}$. Set $Q' = (\hat{Q} \setminus \{\hat{s}\}) \cup \bar{Q} \cup \{s'\}$ and $F' = \hat{F} \setminus \{\hat{s}\}$. Initially, set $R' = \hat{R} \setminus \{a\hat{s} \to q, \hat{s}a \to q \mid a \in \Sigma\}$. Then, extend $R'$ in the following way:
    \begin{enumerate}
        \item for all $as \to q \in R$ and all $sa \to q \in R$, $a \in \Sigma \cup \{\epsilon\}$, $q \in Q$, add $s'a \to \state{q\stl{}}$ and $s'a \to \state{q\str{}}$ to $R'$;\label{theoremproof:ieven_ietwgfa_to_ietwsfa_step_1}
        \item for all $a_1\dots a_na_{n + 1}\dots a_{2n}s \to q \in R$ and all $sa_1\dots a_na_{n + 1}\dots a_{2n} \to q \in R$, $a_i \in \Sigma$, $1 \le i \le 2n$, $n \ge 1$, $q \in Q$, add
        \begin{align*}
            s' &\to \state{a_1\dots a_{n}qa_{n + 1}\dots a_{2n}}, \\
            a_n\state{a_1\dots a_{n}qa_{n + 1}\dots a_{2n}} &\to \state{a_1\dots a_{n - 1}qa_{n + 1}\dots a_{2n}}, \\
            \state{a_1\dots a_{n - 1}qa_{n + 1}\dots a_{2n}}a_{n + 1} &\to \state{a_1\dots a_{n - 1}qa_{n + 2}\dots a_{2n}}, \\
            a_{n - 1}\state{a_1\dots a_{n - 1}qa_{n + 2}\dots a_{2n}} &\to \state{a_1\dots a_{n - 2}qa_{n + 2}\dots a_{2n}}, \\
            &\vdotswithin{\to} \\
            \state{qa_{2n}}a_{2n} &\to \state{q\stl{}}, \\
            \state{a_1\dots a_{n}qa_{n + 1}\dots a_{2n}}a_{n + 1} &\to \state{a_1\dots a_{n}qa_{n + 2}\dots a_{2n}}, \\
            a_{n}\state{a_1\dots a_{n}qa_{n + 2}\dots a_{2n}} &\to \state{a_1\dots a_{n - 1}qa_{n + 2}\dots a_{2n}}, \\
            \state{a_1\dots a_{n - 1}qa_{n + 2}\dots a_{2n}}a_{n + 2} &\to \state{a_1\dots a_{n - 1}qa_{n + 3}\dots a_{2n}}, \\
            &\vdotswithin{\to} \\
            a_1\state{a_1q} &\to \state{q\str{}}
        \end{align*}
        to $R'$;\label{theoremproof:ieven_ietwgfa_to_ietwsfa_step_2}
        \item for all $a_0\dots a_na_{n + 1}\dots a_{2n}s \to q \in R$ and all $sa_0\dots a_na_{n + 1}\dots a_{2n} \to q \in R$, $a_i \in \Sigma$, $0 \le i \le 2n$, $n \ge 1$, $q \in Q$, add
        \begin{align*}
            s'a_n &\to \state{a_0\dots a_{n - 1}qa_{n + 1}\dots a_{2n}}, \\
            a_{n - 1}\state{a_0\dots a_{n - 1}qa_{n + 1}\dots a_{2n}} &\to \state{a_0\dots a_{n - 2}qa_{n + 1}\dots a_{2n}}, \\
            \state{a_0\dots a_{n - 2}qa_{n + 1}\dots a_{2n}}a_{n + 1} &\to \state{a_0\dots a_{n - 2}qa_{n + 2}\dots a_{2n}}, \\
            &\vdotswithin{\to} \\
            \state{qa_{2n}}a_{2n} &\to \state{q\stl{}}, \\
            \state{a_0\dots a_{n - 1}qa_{n + 1}\dots a_{2n}}a_{n + 1} &\to \state{a_0\dots a_{n - 1}qa_{n + 2}\dots a_{2n}}, \\
            a_{n - 1}\state{a_0\dots a_{n - 1}qa_{n + 2}\dots a_{2n}} &\to \state{a_0\dots a_{n - 2}qa_{n + 2}\dots a_{2n}}, \\
            &\vdotswithin{\to} \\
            a_0\state{a_0q} &\to \state{q\str{}}
        \end{align*}
        to $R'$.\label{theoremproof:ieven_ietwgfa_to_ietwsfa_step_3}
    \end{enumerate}

    \emph{Basic Idea.} $M'$ simulates any initialized even computation in $M$ by a sequence of moves, the first of which reads at most one symbol, while all the remaining moves read exactly one symbol at a~time and can, in fact, always be made in such a way that they form an even computation. To explain step \iref{theoremproof:ieven_ietwgfa_to_ietwsfa_step_1}, simply assume that $M$ performs the first move of an initialized even computation according to a rule of the form $as \to q$ or $sa \to q$. Then, this step introduces $s'a \to \state{q\stl{}}$ and $s'a \to \state{q\str{}}$ into $R'$. Clearly, by applying one of these rules, $M'$ simulates the first move of the initialized even computation in $M$. Notice that both of the newly introduced rules, $s'a \to \state{q\stl{}}$ and $s'a \to \state{q\str{}}$, are right because, by the definition of initialized even computation, there are no restrictions based on the direction of the first move of this computation. To explain step \iref{theoremproof:ieven_ietwgfa_to_ietwsfa_step_2}, assume that $M$ performs the first move of an initialized even computation according to a rule of the form $a_1\dots a_na_{n + 1}\dots a_{2n}s \to q$ or $sa_1\dots a_na_{n + 1}\dots a_{2n} \to q$. Consider the sequence of rules introduced into $R'$ in step \iref{theoremproof:ieven_ietwgfa_to_ietwsfa_step_2}. Observe that once $M'$ applies its first rule, it has to continue by applying the rules from this sequence until it reaches either state $\state{q\stl{}}$ or $\state{q\str{}}$. During this process, $M'$ reads the string $a_1\dots a_na_{n + 1}\dots a_{2n}$. Thus, it simulates the first move of the initialized even computation in $M$. Notice that the first rule, $s' \to \state{a_1\dots a_{n}qa_{n + 1}\dots a_{2n}}$, is an $\epsilon$-rule. This is because the sequence $a_1\dots a_na_{n + 1}\dots a_{2n}$ contains an even number of symbols, but any initialized even computation always consists of an odd number of moves. Step \iref{theoremproof:ieven_ietwgfa_to_ietwsfa_step_3} is analogous to step \iref{theoremproof:ieven_ietwgfa_to_ietwsfa_step_2}, except that the first rule of the introduced sequence is of the form $s'a_n \to \state{a_0\dots a_{n - 1}qa_{n + 1}\dots a_{2n}}$, as $a_0\dots a_na_{n + 1}\dots a_{2n}$ consists of an odd number of symbols. The rest of an initialized even computation in $M$, more precisely, its even part, is simulated by $M'$ in the same way as by $\hat{M}$ (for details see the proof of Theorem \ref{theorem:ie2gfa_even_in_e-free_ie2sfa_even}).

    Now, we establish $L(M')_\ieven{} = L(M)_\ieven{}$ formally. From the proof of Theorem~\ref{theorem:ie2gfa_even_in_e-free_ie2sfa_even}, it follows that for all $p, q \in Q$ and $u, v \in \Sigma^*$, $u\state{q\stl{}}v \Ra^*_\even{} \state{p\stl{}}$ or $u\state{q\str{}}v \Ra^*_\even{} \state{p\str{}}$ in $M'$, where $\state{p\stl{}}, \state{p\str{}} \in F'$, iff there is $f \in F$ such that $uqv \Ra^*_\even{} p \Ra^*_\even{} f$ in $M$. Then, according to steps \iref{theoremproof:ieven_ietwgfa_to_ietwsfa_step_1}, \iref{theoremproof:ieven_ietwgfa_to_ietwsfa_step_2}, and \iref{theoremproof:ieven_ietwgfa_to_ietwsfa_step_3} of the construction technique of $M'$, the following holds:
    \begin{enumerate}[label=(\roman*)]
        \item
        $us'av \Ra u\state{q\stl{}}v \Ra^*_\even{} \state{p\stl{}}$ or $us'av \Ra u\state{q\str{}}v \Ra^*_\even{} \state{p\str{}}$ in $M'$, where $\state{p\stl{}}, \state{p\str{}} \in F'$, iff there is $f \in F$ such that $uasv \Ra uqv \Ra^*_\even{} p \Ra^*_\even{} f$ or $usav \Ra uqv \Ra^*_\even{} p \Ra^*_\even{} f$ in $M$;
        
        \item
        for all $n \ge 1$,
        \begin{align*}
            ua_1\dots a_ns'a_{n + 1}\dots a_{2n}v &\Ra ua_1\dots a_n\state{a_1\dots a_{n}qa_{n + 1}\dots a_{2n}}a_{n + 1}\dots a_{2n}v \\
            &\Ra ua_1\dots a_{n - 1}\state{a_1\dots a_{n - 1}qa_{n + 1}\dots a_{2n}}a_{n + 1}\dots a_{2n}v \\
            &\Ra ua_1\dots a_{n - 1}\state{a_1\dots a_{n - 1}qa_{n + 2}\dots a_{2n}}a_{n + 2}\dots a_{2n}v \\
            &\Ra ua_1\dots a_{n - 2}\state{a_1\dots a_{n - 2}qa_{n + 2}\dots a_{2n}}a_{n + 2}\dots a_{2n}v \\
            &\vdotswithin{\to} \\
            &\Ra u\state{qa_{2n}}a_{2n}v \Ra u\state{q\stl{}}v \Ra^*_\even{} \state{p\stl{}}
        \end{align*}
        or
        \begin{align*}
            ua_1\dots a_ns'a_{n + 1}\dots a_{2n}v &\Ra ua_1\dots a_n\state{a_1\dots a_{n}qa_{n + 1}\dots a_{2n}}a_{n + 1}\dots a_{2n}v \\
            &\Ra ua_1\dots a_n\state{a_1\dots a_{n}qa_{n + 2}\dots a_{2n}}a_{n + 2}\dots a_{2n}v \\
            &\Ra ua_1\dots a_{n - 1}\state{a_1\dots a_{n - 1}qa_{n + 2}\dots a_{2n}}a_{n + 2}\dots a_{2n}v \\
            &\Ra ua_1\dots a_{n - 1}\state{a_1\dots a_{n - 1}qa_{n + 3}\dots a_{2n}}a_{n + 3}\dots a_{2n}v \\
            &\vdotswithin{\to} \\
            &\Ra ua_1\state{a_1q}v \Ra u\state{q\str{}}v \Ra^*_\even{} \state{p\str{}}
        \end{align*}
        in $M'$, where $\state{p\stl{}}, \state{p\str{}} \in F'$, iff there is $f \in F$ such that $$ua_1\dots a_{2n}sv \Ra uqv \Ra^*_\even{} p \Ra^*_\even{} f \text{ or } usa_1\dots a_{2n}v \Ra uqv \Ra^*_\even{} p \Ra^*_\even{} f$$ in $M$;
        
        \item
        for all $n \ge 1$,
        \begin{align*}
            ua_0\dots s'a_na_{n + 1}\dots a_{2n}v &\Ra ua_0\dots a_{n - 1}\state{a_0\dots a_{n - 1}qa_{n + 1}\dots a_{2n}}a_{n + 1}\dots a_{2n}v \\
            &\Ra ua_0\dots a_{n - 2}\state{a_0\dots a_{n - 2}qa_{n + 1}\dots a_{2n}}a_{n + 1}\dots a_{2n}v \\
            &\Ra ua_0\dots a_{n - 2}\state{a_0\dots a_{n - 2}qa_{n + 2}\dots a_{2n}}a_{n + 2}\dots a_{2n}v \\
            &\vdotswithin{\to} \\
            &\Ra u\state{qa_{2n}}a_{2n}v \Ra u\state{q\stl{}}v \Ra^*_\even{} \state{p\stl{}}
        \end{align*}
        or
        \begin{align*}
            ua_0\dots s'a_na_{n + 1}\dots a_{2n}v &\Ra ua_1\dots a_{n - 1}\state{a_0\dots a_{n - 1}qa_{n + 1}\dots a_{2n}}a_{n + 1}\dots a_{2n}v \\
            &\Ra ua_1\dots a_{n - 1}\state{a_0\dots a_{n - 1}qa_{n + 2}\dots a_{2n}}a_{n + 2}\dots a_{2n}v \\
            &\Ra ua_1\dots a_{n - 2}\state{a_0\dots a_{n - 2}qa_{n + 2}\dots a_{2n}}a_{n + 2}\dots a_{2n}v \\
            &\vdotswithin{\to} \\
            &\Ra ua_0\state{a_0q}v \Ra u\state{q\str{}}v \Ra^*_\even{} \state{p\str{}}
        \end{align*}
        in $M'$, where $\state{p\stl{}}, \state{p\str{}} \in F'$, iff there is $f \in F$ such that $$ua_0\dots a_{2n}sv \Ra uqv \Ra^*_\even{} p \Ra^*_\even{} f \text{ or } usa_0\dots a_{2n}v \Ra uqv \Ra^*_\even{} p \Ra^*_\even{} f$$ in $M$.
    \end{enumerate}
    Clearly, from states of the form $\state{q\stl{}}$ and $\state{q\str{}}$, where $q \in Q$, $M'$ can only make left and right moves, respectively. Thus, by the definition of initialized even computation, we can express the previous equivalences as follows:
    \begin{enumerate}[label=(\roman*)]
        \item
        $us'av \Ra^*_\ieven{} \state{p\stl{}}$ or $us'av \Ra^*_\ieven{} \state{p\str{}}$ in $M'$, where $\state{p\stl{}}, \state{p\str{}} \in F'$, iff there is $f \in F$ such that $uasv \Ra^*_\ieven{} p \Ra^*_\even{} f$ or $usav \Ra^*_\ieven{} p \Ra^*_\even{} f$ in $M$;
        
        \item
        for all $n \ge 1$, $$ua_1\dots a_ns'a_{n + 1}\dots a_{2n}v \Ra^*_\ieven{} \state{p\stl{}} \text{ or } ua_1\dots a_ns'a_{n + 1}\dots a_{2n}v \Ra^*_\ieven{} \state{p\str{}}$$ in $M'$, where $\state{p\stl{}}, \state{p\str{}} \in F'$, iff there is $f \in F$ such that $$ua_1\dots a_{2n}sv \Ra^*_\ieven{} p \Ra^*_\even{} f \text{ or } usa_1\dots a_{2n}v \Ra^*_\ieven{} p \Ra^*_\even{} f$$ in $M$;
        
        \item
        for all $n \ge 1$, $$ua_0\dots s'a_na_{n + 1}\dots a_{2n}v \Ra^*_\ieven{} \state{p\stl{}} \text{ or } ua_0\dots s'a_na_{n + 1}\dots a_{2n}v \Ra^*_\ieven{} \state{p\str{}}$$ in $M'$, where $\state{p\stl{}}, \state{p\str{}} \in F'$, iff there is $f \in F$ such that $$ua_0\dots a_{2n}sv \Ra^*_\ieven{} p \Ra^*_\even{} f \text{ or } usa_0\dots a_{2n}v \Ra^*_\ieven{} p \Ra^*_\even{} f$$ in $M$.
    \end{enumerate}
    Based on the above information, we can safely conclude that for all $w_1, w_2 \in \Sigma^*$ and $p \in Q$, $w_1s'w_2 \Ra^*_\ieven{} \state{p\stl{}}$ or $w_1s'w_2 \Ra^*_\ieven{} \state{p\str{}}$, where $\state{p\stl{}}, \state{p\str{}} \in F'$, iff there is $f \in F$ such that $w_1sw_2 \Ra^*_\ieven{} p \Ra^*_\even{} f$ in $M$. Hence, $L(M')_\ieven{} = L(M)_\ieven{}$.

    Obviously, $L(M')_\ieven{} \subseteq L(M')$. Observe, however, that by the construction of $M'$, there is no $w \in \Sigma^*$ such that $w \in L(M') \setminus L(M')_\ieven{}$, so $L(M') \subseteq L(M')_\ieven{}$. In addition, notice that $s'$ never occurs on the right side of any rule, all $\epsilon$-rules have $s'$ on their left-hand side. Therefore, Theorem \ref{theorem:ieven_ietwgfa_to_ietwsfa} holds.
\end{proof}

\begin{lemma}\label{lemma:ieven_ietwgfa_in_elg}
    For every \ietwgfa{} $M$, there is an \eling{} $G$ such that $L(G) = L(M)_\ieven{}$.
\end{lemma}

\begin{proof}
    Let $M = (Q, \Sigma, R, s, F)$ be an \ietwgfa{}. From $M$, we next construct an \eling{} $G = (N, T, P, S)$ such that $L(G) = L(M)_\ieven{}$. Introduce a new symbol $S$---the start nonterminal of $G$. Set $N' = \{\state{qd} \mid q \in Q, d \in \{\stl{}, \str{}\}\}$. Assume that $S \notin N'$. Set $N = N' \cup \{S\}$ and $T = \Sigma$. $P$ is then constructed as follows:
    \begin{enumerate}
        \item for all $f \in F$, add $S \to \state{f\stl{}}$ and $S \to \state{f\str{}}$ to $P$;\label{lemma:ieven_ietwgfa_in_elg_step_1}
        \item if $xs \to q \in R$ or $sx \to q \in R$, where $q \in Q$ and $x \in \Sigma^*$, add $\state{q\stl{}} \to x$ and $\state{q\str{}} \to x$ to $P$;\label{lemma:ieven_ietwgfa_in_elg_step_2}
        \item if $xq \to p, py \to o \in R$, where $o, p, q \in Q$, $x, y \in \Sigma^*$, and $\strlen{x} = \strlen{y}$, add $\state{o\stl{}} \to x\state{q\stl{}}y$ to $P$;\label{lemma:ieven_ietwgfa_in_elg_step_3}
        \item if $qy \to p, xp \to o \in R$, where $o, p, q \in Q$, $x, y \in \Sigma^*$, and $\strlen{x} = \strlen{y}$, add $\state{o\str{}} \to x\state{q\str{}}y$ to $P$;\label{lemma:ieven_ietwgfa_in_elg_step_4}
    \end{enumerate}

    \emph{Basic Idea.} $G$ simulates any initialized even computation of $M$ in reverse. It starts by generating a nonterminal of the form $\state{f\stl{}}$ or $\state{f\str{}}$, which corresponds to a~final state $f \in F$ (see step \iref{lemma:ieven_ietwgfa_in_elg_step_1}). After this initial derivation step, $G$ simulates every two-move even computation made by $M$ according to two consecutive rules of the form $xq \to p$ and $py \to o$, where $o, p, q \in Q$ and $x, y \in \Sigma^*$, by using a rule of the form $\state{o\stl{}} \to x\state{q\stl{}}y$ (see step \iref{lemma:ieven_ietwgfa_in_elg_step_3}). Notice that the first rule, $xq \to p$, is a left rule. Step \iref{lemma:ieven_ietwgfa_in_elg_step_4} is analogous to step \iref{lemma:ieven_ietwgfa_in_elg_step_3}, except that the first of the two consecutive rules is a right rule. As can be seen, the even part of an initialized even computation of $M$ is simulated in $G$ by a derivation over nonterminals of the form $\state{q\stl{}}$, where $q \in Q$, if it starts with a~left rule; otherwise, it is simulated by a derivation over nonterminals of the form $\state{q\str{}}$. The simulation process is completed by applying a rule of the form $\state{q\stl{}} \to x$ or $\state{q\str{}} \to x$, where $q \in Q$ and $x \in \Sigma^*$. Thus, the symbol sequence read by the first move of an initialized even computation of $M$ is generated, and a string of terminals in $G$ is obtained (see step~\iref{lemma:ieven_ietwgfa_in_elg_step_2}).

    Let us now establish $L(G) = L(M)_\ieven{}$ formally. We start by proving the following two equivalences.
    
    For all $u, v \in \Sigma^*$ and $p, q \in Q$, 
    \begin{equation}\label{claim:lemma:ieven_ietwgfa_in_elg_A}
        \state{q\stl{}} \Ra^* u\state{p\stl{}}v \text{ in } G \iff upv \Ra^*_\even{} q \text{ in } M,
    \end{equation}
    where $upv \Ra^*_\even{} q$ starts with a left move (unless it is an empty sequence of moves).

    First, we establish the \emph{only-if} part of equivalence \eqref{claim:lemma:ieven_ietwgfa_in_elg_A}. By induction on the number of derivation steps $i \ge 0$, we prove that $\state{q\stl{}} \Ra^i u\state{p\stl{}}v \text{ in } G$ implies $upv \Ra^*_\even{} q \text{ in } M$, where $upv \Ra^*_\even{} q$ starts with a left move (or consists of no moves). Let $i = 0$, so $\state{q\stl{}} \Ra^0 u\state{p\stl{}}v$ in $G$. Then, $q = p$ and $uv = \epsilon$. Since $q \Ra^0_\even{} q$ in $M$, the basis holds true. Assume that the implication holds for all derivations consisting of no more than $j$ steps, for some $j \in \mathbb{N}_0$. Consider any derivation of the form $\state{q\stl{}} \Ra^{j + 1} u\state{p\stl{}}v$ in $G$. Let this derivation start with the application of a rule of the form $$\state{q\stl{}} \to x\state{o\stl{}}y$$ from $P$, where $o \in Q$, $x, y \in \Sigma^*$, and $\strlen{x} = \strlen{y}$. Thus, we can express $\state{q\stl{}} \Ra^{j + 1} u\state{p\stl{}}v$ as $$\state{q\stl{}} \Ra x\state{o\stl{}}y \Ra^j xu'\state{p\stl{}}v'y$$ in $G$, where $xu' = u$ and $v'y = v$. By the induction hypothesis, $u'pv' \Ra^*_\even{} o$ in $M$, and this computation starts with a left move (or consists of no moves at all). Step~\iref{lemma:ieven_ietwgfa_in_elg_step_3} constructs $\state{q\stl{}} \to x\state{o\stl{}}y \in P$ from two consecutive rules $xo \to t, ty \to q \in R$, for some $t \in Q$, so $$xu'pv'y \Ra^*_\even{} xoy \Ra ty \Ra q$$ in $M$. Since $xu' = u$, $v'y = v$, and $\strlen{x} = \strlen{y}$, taking into account the properties of $u'pv' \Ra^*_\even{} o$, it follows that $upv \Ra^*_\even{} q$ in $M$. As we can see, $upv \Ra^*_\even{} q$ starts with a~left move. Thus, the induction step is completed.

    Next, we establish the \emph{if} part of equivalence \eqref{claim:lemma:ieven_ietwgfa_in_elg_A}. By induction on the number of moves $i \ge 0$, we prove that $upv \Ra^i_\even{} q$, which starts with a left move (or consists of no moves), in $M$ implies $\state{q\stl{}} \Ra^* u\state{p\stl{}}v \text{ in } G$. For $i = 0$, $upv \Ra^0_\even{} q$ occurs in $M$ only for $p = q$ and $uv = \epsilon$. In $G$, clearly, $\state{q\stl{}} \Ra^0 \state{q\stl{}}$. For $i = 1$, $upv \Ra^1_\even{} q$ never occurs in $M$ since, by Definition \ref{def:computational_restrictions}, every even computation is supposed to have an even number of moves; however, $upv \Ra^1_\even{} q$ has one move only. Thus, the basis holds true. Assume that the implication holds for all computations consisting of no more than $j$ moves, for some $j \in \mathbb{N}_0$. Let $upv \Ra^{j + 2}_\even{} q$ in $M$, and let this computation end with the application of two consecutive rules of the form $$xo \to t \text{ and } ty \to q$$ from $R$, where $o, t \in Q$, $x, y \in \Sigma^*$, and $\strlen{x} = \strlen{y}$. Express $upv \Ra^{j + 2}_\even{} q$ as $$xu'pv'y \Ra^j_\even{} xoy \Ra ty \Ra q$$ in $M$, where $xu' = u$ and $v'y = v$. Observe that $u'pv' \Ra^j_\even{} o$ starts with a left move (or consists of no moves at all). Thus, by the induction hypothesis, $\state{o\stl{}} \Ra^* u'\state{p\stl{}}v'$ in $G$. From $xo \to t, ty \to q \in R$, step \iref{lemma:ieven_ietwgfa_in_elg_step_3} constructs $\state{q\stl{}} \to x\state{o\stl{}}y \in P$, so $G$ can make $$\state{q\stl{}} \Ra x\state{o\stl{}}y \Ra^* xu'\state{p\stl{}}v'y.$$ Because $xu' = u$ and $v'y = v$, it follows that $\state{q\stl{}} \Ra^* u\state{p\stl{}}v$ in $G$. Thus, the induction step is completed, and equivalence \eqref{claim:lemma:ieven_ietwgfa_in_elg_A} holds.

    For all $u, v \in \Sigma^*$ and $p, q \in Q$,
    \begin{equation}\label{claim:lemma:ieven_ietwgfa_in_elg_B}
        \state{q\str{}} \Ra^* u\state{p\str{}}v \text{ in } G \iff upv \Ra^*_\even{} q \text{ in } M,
    \end{equation}
    where $upv \Ra^*_\even{} q$ starts with a right move (unless it is an empty sequence of moves).

    Equivalence \eqref{claim:lemma:ieven_ietwgfa_in_elg_B} can be proved analogously with the proof of equivalence \eqref{claim:lemma:ieven_ietwgfa_in_elg_A}.

    As a consequence of equivalences \eqref{claim:lemma:ieven_ietwgfa_in_elg_A} and \eqref{claim:lemma:ieven_ietwgfa_in_elg_B}, we obtain that for all $u, v \in \Sigma^*$ and $p, q \in Q$, $\state{q\stl{}} \Ra^* u\state{p\stl{}}v \text{ or } \state{q\str{}} \Ra^* u\state{p\str{}}v \text{ in } G \iff upv \Ra^*_\even{} q \text{ in } M$. Next, consider this equivalence for $p \in \{o \mid xs \to o \in R, o \in Q, x \in \Sigma^*\} \cup \{o \mid sx \to o \in R, o \in Q, x \in \Sigma^*\}$. As follows from the construction technique, $G$ starts every derivation by applying a~rule of the form $S \to \state{fd}$, where $f \in F$ and $d \in \{\stl{}, \str{}\}$, and ends it by applying a~rule of the form $\state{pd} \to x$, where $x \in \Sigma^*$, constructed from $xs \to p \in R$ or $sx \to p \in R$ by step \iref{lemma:ieven_ietwgfa_in_elg_step_2}. Consequently, $S \Ra \state{f\stl{}} \Ra^* u\state{p\stl{}}v \Ra uxv$ or $S \Ra \state{f\str{}} \Ra^* u\state{p\str{}}v \Ra uxv$ in $G$ iff $uxsv \Ra upv \Ra^*_\even{} f$ or $usxv \Ra upv \Ra^*_\even{} f$ in $M$. Hence, by the definition of initialized even computation, $S \Ra \state{f\stl{}} \Ra^* u\state{p\stl{}}v \Ra uxv$ or $S \Ra \state{f\str{}} \Ra^* u\state{p\str{}}v \Ra uxv$ in $G$ iff $uxsv \Ra^*_\ieven{} f$ or $usxv \Ra^*_\ieven{} f$ in $M$. As a result, $L(G) = L(M)_\ieven{}$, so Lemma \ref{lemma:ieven_ietwgfa_in_elg} holds.
\end{proof}

\begin{lemma}\label{lemma:elg_in_ieven_ietwgfa}
    For every \eling{} $G$, there is an \ietwgfa{} $M$ such that $L(M)_\ieven = L(G)$.
\end{lemma}

\begin{proof}
    Let $G = (N, T, P, S)$ be an \eling{} and $M = (Q, \Sigma, R, s, F)$ be an \ietwgfa{} constructed from $G$ using the technique described in the proof of Lemma \ref{lemma:lg_in_ietwgfa}. As follows from the proof of Lemma \ref{lemma:lg_in_ietwgfa}, $uszv \Ra uCv \Ra^* S \text{ in } M \iff S \Ra^* uCv \Ra uzv \text{ in } G$ for all $C \in N$ and $u, v, z \in T^*$. According to the technique used for the construction of $M$, for each $A \to xBy \in P$, where $A, B \in N$ and $x, y \in T^*$, there are two consecutive rules $xA \to \state{A \to xBy}, \state{A \to xBy}y \to B \in R$, which are always applied one immediately after the other in the given order. Thus, $uCv \Ra^* S$ consists of an even number of moves and is alternating, so $uCv \Ra^*_\alt{} S$. Furthermore, since $G$ is even, $\strlen{x} = \strlen{y}$ always holds; hence, $uCv \Ra^*_\alt{} S$ is an even computation, so $uCv \Ra^*_\even{} S$. Consequently, $uszv \Ra uCv \Ra^*_\even{} S \text{ in } M \iff S \Ra^* uCv \Ra uzv \text{ in } G$. Hence, by the definition of initialized even computation, $uszv \Ra^*_\ieven{} S \text{ in } M \iff S \Ra^* uCv \Ra uzv \text{ in } G$, so $L(M)_\ieven{} = L(G)$. Therefore, Lemma \ref{lemma:elg_in_ieven_ietwgfa} holds.
\end{proof}

\begin{theorem}
    $\lfam{\epsilon}{\ietwgfa{}}_\ieven = \lfam{\epsilon}{\eling{}}$.
\end{theorem}

\begin{proof}
    As $\lfam{\epsilon}{\ietwgfa{}}_\ieven \subseteq \lfam{\epsilon}{\eling{}}$ follows from Lemma \ref{lemma:ieven_ietwgfa_in_elg} and $\lfam{\epsilon}{\eling{}} \subseteq \lfam{\epsilon}{\ietwgfa{}}_\ieven$ from Lemma \ref{lemma:elg_in_ieven_ietwgfa}, the identity $\lfam{\epsilon}{\ietwgfa{}}_\ieven = \lfam{\epsilon}{\eling{}}$ clearly holds.
\end{proof}

\begin{theorem}
    $\lfam{\epsilon}{\ietwgfa{}}_\even{} \subset \lfam{\epsilon}{\ietwgfa{}}_\ieven{}$.
\end{theorem}

\begin{proof}
    First, we prove $\lfam{\epsilon}{\ietwgfa{}}_\even{} \subseteq \lfam{\epsilon}{\ietwgfa{}}_\ieven{}$. Consider any \ietwgfa{} $M = (Q, \Sigma, R, s, F)$. From $M$, construct the \ietwgfa{} $M' = (Q \cup \{s'\}, \Sigma, R \cup \{s' \to s\}, s', F)$, where $s' \notin Q$. Clearly, $L(M')_\ieven{} = L(M)_\even{}$, so $\lfam{\epsilon}{\ietwgfa{}}_\even{} \subseteq \lfam{\epsilon}{\ietwgfa{}}_\ieven{}$.
    
    To prove that $\lfam{\epsilon}{\ietwgfa{}}_\ieven{} \setminus \lfam{\epsilon}{\ietwgfa{}}_\even{} \neq \emptyset$, consider the language $\{a\}$. Clearly, $\{a\} \in \lfam{\epsilon}{\ietwgfa{}}_\ieven{}$. However, by Theorem \ref{theorem:ie2gfa_even_in_even}, $\{a\} \notin \lfam{\epsilon}{\ietwgfa{}}_\even{}$. Thus, $\lfam{\epsilon}{\ietwgfa{}}_\even{} \subset \lfam{\epsilon}{\ietwgfa{}}_\ieven{}$.
\end{proof}

\begin{theorem}\label{theorem:efree_ie2sfa_ieven_prop_in_odd_lfam}
    $\lfam{}{\ietwsfa{}}_\ieven{} \subset \lfam{}{\oddlfam{}}$.
\end{theorem}

\begin{proof}
    According to Definition \ref{def:computational_restrictions}, each initialized even computation consists of an odd number of moves. Therefore, no $\epsilon$-free \ietwsfa{} can ever accept any even-length string in this way, since it always reads exactly one symbol per move. Consequently, each language in $\lfam{}{\ietwsfa{}}_\ieven{}$ consists of odd-length strings only, so $\lfam{}{\ietwsfa{}}_\ieven{} \subset \lfam{}{\oddlfam{}}$.
\end{proof}

\section{Input-Related Restrictions}\label{sec:input_related_restrictions}

This chapter studies input-related restrictions of \ietwgfa{}s. More specifically, it investigates the power of these automata working under the assumption that their input strings or their parts belong to some prescribed language families, such as the regular language family. Theorems \ref{theorem:inputABregular_results_in_linear} and \ref{theorem:theorem:inputsplitAregular_results_in_linear} show that regular-based input restrictions give rise to no increase in the power of \ietwgfa{}s. Theorems \ref{theorem:inputAregularBfinal_results_in_reg} and \ref{theorem:inputAregBregCreg_readingB_results_in_reg} demonstrate that regular-based input restrictions even lead to a decrease in the power to that of ordinary \fa{}s. These results are of some interest only when compared to the investigation of similar restrictions placed upon other rewriting systems, in which these restrictions give rise to a significant increase in their power. For instance, most selective grammars with regular-based selectors, which restrict the rewritten strings, are as strong as Turing machines (see Chapter 10 in \cite{REGULATED_REWRITING_DASSOW_2011} for a summary). In view of this increase in power in terms of other rewriting mechanisms, at a glance, we might hastily expect analogical results in terms of \ietwgfa{}s, but the present section demonstrates that this is not the case.

\begin{theorem}\label{theorem:inputABregular_results_in_linear}
    Let $M = (Q, \Sigma, R, s, F)$ be an $\epsilon$-free \ietwsfa{} and $A, B \subseteq \Sigma^*$ be regular. Then, there is an $\epsilon$-free \ietwsfa{} $M'$ such that $L(M') = \left\{uv \mid usv \Ra^* f \text{ in } M, f \in F, u \in A, v \in B\right\}$, so $L(M')$ is linear.
\end{theorem}

\begin{proof}
    Let $M = (Q, \Sigma, R, s, F)$ be an $\epsilon$-free \ietwsfa{} and $A, B \subseteq \Sigma^*$ be regular. Let $A = L(M_1)$ and $B = L(M_2)$, where $M_i = (Q_i, \Sigma, R_i, s_i, F_i)$ is an $\epsilon$-free \fa{} for all $i \in \{1, 2\}$. From $M$, $M_1$, and $M_2$, we construct an $\epsilon$-free \ietwsfa{} $M' = (Q', \Sigma, R', s', F')$ such that $L(M') = \{uv \mid usv \Ra^* f \text{ in } M, f \in F, u \in A, v \in B\}$. Introduce a new symbol $s'$---the start state of $M'$. Set $\hat{Q} = \{\state{qq_1q_2} \mid q \in Q, q_i \in Q_i, i \in \{1, 2\}\}$. Assume that $s' \notin \hat{Q}$. Set $Q' = \hat{Q} \cup \{s'\}$ and $F' = \{\state{fs_1f_2} \mid f \in F, f_2 \in F_2\}$. Initially, set $R' = \emptyset$. Then, extend $R'$ in the following way:
    \begin{enumerate}
        \item for all $f_1 \in F_1$, add $s' \to \state{sf_1s_2}$ to $R'$;\label{theoremproof:inputABregular_results_in_linear_step_1}
        \item if $ap \to q \in R$ and $q_1a \to p_1 \in R_1$, where $p, q \in Q$, $p_1, q_1 \in Q_1$, and $a \in \Sigma$, add $a\state{pp_1q_2} \to \state{qq_1q_2}$ to $R'$ for all $q_2 \in Q_2$.\label{theoremproof:inputABregular_results_in_linear_step_2}
        \item if $pa \to q \in R$ and $p_2a \to q_2 \in R_2$, where $p, q \in Q$, $p_2, q_2 \in Q_2$, and $a \in \Sigma$, add $\state{pq_1p_2}a \to \state{qq_1q_2}$ to $R'$ for all $q_1 \in Q_1$.\label{theoremproof:inputABregular_results_in_linear_step_3}
    \end{enumerate}

    \emph{Basic Idea.} $M'$, in effect, works in a two-directional way. To the right, it simulates a computation made by $M$ and, simultaneously, a computation made by $M_2$ (see step \iref{theoremproof:inputABregular_results_in_linear_step_3}). To the left, it simulates a computation made by $M$ and, simultaneously, a computation made by $M_1$ in reverse (see step \iref{theoremproof:inputABregular_results_in_linear_step_2}). Consider step \iref{theoremproof:inputABregular_results_in_linear_step_1} to see that $M'$ accepts its input if and only if all the three automata; $M$, $M_1$, and $M_2$; accept their inputs as well, so  $L(M') = \{uv \mid usv \Ra^* f \text{ in } M, f \in F, u \in A, v \in B\}$.

    Let us now establish $L(M') = \left\{uv \mid usv \Ra^* f \text{ in } M, f \in F, u \in A, v \in B\right\}$ formally. We start by proving the following equivalence.

    For all $u,v \in \Sigma^*$, $p, q \in Q$, $p_1, q_1 \in Q_1$, and $p_2, q_2 \in Q_2$,
    \begin{equation}\label{claim:theorem:inputABregular_results_in_linear}
        u\state{pp_1p_2}v \Ra^* \state{qq_1q_2} \text{ in } M' \iff upv \Ra^* q \text{ in } M, q_1u \Ra^* p_1 \text{ in } M_1,\text{and } p_2v \Ra^* q_2 \text{ in } M_2.
    \end{equation}

    First, we establish the \emph{only if} part of this equivalence. By induction on the number of moves $i \ge 0$, we prove that $u\state{pp_1p_2} v \Ra^i \state{qq_1q_2}$ in $M' \implies upv \Ra^* q$ in $M$, $q_1u \Ra^* p_1$ in $M_1$, and $p_2v \Ra^* q_2$ in $M_2$. Let $i = 0$, so $u\state{pp_1p_2}v \Ra^0 \state{qq_1q_2}$ in $M'$. Then, $p = q$, $p_1 = q_1$, $p_2 = q_2$, and $uv = \epsilon$. Clearly, $p \Ra^0 p$ in $M$, $p_1 \Ra^0 p_1$ in $M_1$, and $p_2 \Ra^0 p_2$ in $M_2$, so the basis holds true. Assume that the implication holds for all computations consisting of no more than $j$ moves in $M'$, for some $j \in \mathbb{N}_0$. Consider any computation of the form $u\state{pp_1p_2}v \Ra^{j + 1} \state{qq_1q_2}$ in $M'$. Let this computation start with the application of a rule of the form $$a\state{pp_1p_2} \to \state{oo_1p_2}$$ from $R'$, where $o \in Q$, $o_1 \in Q_1$, and $a \in \Sigma$. Thus, we can express $u\state{pp_1p_2}v \Ra^{j + 1} \state{qq_1q_2}$ as $$u'a\state{pp_1p_2}v \Ra u'\state{oo_1p_2}v \Ra^j \state{qq_1q_2}$$ in $M'$, where $u'a = u$. By the induction hypothesis, $u'ov \Ra^* q$ in $M$, $q_1u' \Ra^* o_1$ in $M_1$, and $p_2v \Ra^* q_2$ in $M_2$. Since step \iref{theoremproof:inputABregular_results_in_linear_step_2} constructs $a\state{pp_1p_2} \to \state{oo_1p_2} \in R'$ from $o_1a \to p_1 \in R_1$ and $ap \to o \in R$, $$u'apv \Ra u'ov \Ra^* q$$ in $M$ and $$q_1u'a \Ra^* o_1a \Ra p_1$$ in $M_1$. Because $u'a = u$, $upv \Ra^* q$ in $M$ and $q_1u \Ra^* p_1$ in $M_1$. In the case when the derivation $u\state{pp_1p_2}v \Ra^{j + 1} \state{qq_1q_2}$ in $M'$ starts with the application of a rule of the form $\state{pp_1p_2}a \to \state{op_1o_2}$ from $R'$, where $o \in Q$, $o_2 \in Q_2$, and $a \in \Sigma$, we can proceed analogously. Thus, the induction step is completed.

    Next, we establish the \emph{if} part of equivalence \eqref{claim:theorem:inputABregular_results_in_linear}, so we show that $upv \Ra^i q$ in $M$, $q_1u \Ra^j p_1$ in $M_1$, and $p_2v \Ra^k q_2$ in $M_2$, where $i = j + k$, implies $u\state{pp_1p_2} v \Ra^* \state{qq_1q_2}$ in $M'$ by induction on the number of moves $i \ge 0$. Let $i = 0$, so $j = 0$, $k = 0$, $upv \Ra^0 q$ in $M$, $q_1u \Ra^0 p_1$ in $M_1$, and $p_2v \Ra^0 q_2$ in $M_2$. Then, $p = q$, $p_1 = q_1$, $p_2 = q_2$, and $uv = \epsilon$. Since $\state{pp_1p_2} \Ra^0 \state{pp_1p_2}$ in $M'$, the basis holds true. Assume that the implication holds for all computations consisting of no more than $l$ moves in $M$, for some $l \in \mathbb{N}_0$. Consider any $upv \Ra^{l + 1} q$ in $M$, $q_1u \Ra^{m + 1} p_1$ in $M_1$, and $p_2v \Ra^{n} q_2$ in $M_2$, where $m + n = l$. Let $upv \Ra^{l + 1} q$ in $M$ start with the application of a rule of the form $$ap \to o$$ from $R$ and $q_1u \Ra^{m + 1} p_1$ in $M_1$ end with the application of a rule of the form $$o_1a \to p_1$$ from $R_1$, where $o \in Q$, $o_1 \in Q_1$, and $a \in \Sigma$. Now, express $upv \Ra^{l + 1} q$ as $$u'apv \Ra u'ov \Ra^l q$$ in $M$ and $q_1u \Ra^{m + 1} p_1$ as $$q_1u'a \Ra^m o_1a \Ra p_1$$ in $M_1$, where $u'a = u$. By the induction hypothesis, we have $u'\state{oo_1p_2}v \Ra^* \state{qq_1q_2}$ in $M'$. From $ap \to o \in R$ and $o_1a \to p_1 \in R_1$, step \iref{theoremproof:inputABregular_results_in_linear_step_2} constructs $a\state{pp_1p_2} \to \state{oo_1p_2} \in R'$. Thus, $M'$ makes $$u'a\state{pp_1p_2}v \Ra u'\state{oo_1p_2}v \Ra^* \state{qq_1q_2}.$$ Since $u'a = u$, $u\state{pp_1p_2}v \Ra^* \state{qq_1q_2}$ in $M'$. Next, consider any $upv \Ra^{l + 1} q$ in $M$, $q_1u \Ra^{m} p_1$ in $M_1$, and $p_2v \Ra^{n + 1} q_2$ in $M_2$, where $m + n = l$. Let $upv \Ra^{l + 1} q$ in $M$ start with the application of a rule of the form $pa \to o$ from $R$ and $p_2v \Ra^{n + 1} q_2$ in $M_2$ start with the application of a rule of the form $p_2a \to o_2$ from $R_2$, where $o \in Q$, $o_2 \in Q_2$, and $a \in \Sigma$. Then, proceed by analogy with the previous~case. Thus, the induction step is completed, and equivalence \eqref{claim:theorem:inputABregular_results_in_linear} holds.

    Consider equivalence \eqref{claim:theorem:inputABregular_results_in_linear} for $p = s$, $q_1 = s_1$, and $p_2 = s_2$. At this point, for all $u, v \in \Sigma^*$, $q \in Q$, $p_1 \in Q_1$, and $q_2 \in Q_2$, $u\state{sp_1s_2}v \Ra^* \state{qs_1q_2}$ in $M' \iff usv \Ra^* q$ in $M$, $s_1u \Ra^* p_1$ in $M_1$, and $s_2v \Ra^* q_2$ in $M_2$. As follows from the construction of $R'$, $M'$ starts every computation by applying a rule of the form $s' \to \state{sf_1s_2}$ with $f_1 \in F_1$. Consequently, $us'v \Ra u\state{sf_1s_2}v \Ra^* \state{qs_1q_2}$ in $M' \iff usv \Ra^* q$ in $M$, $s_1u \Ra^* f_1$ in $M_1$, and $s_2v \Ra^* q_2$ in $M_2$. Considering this equivalence for $q = f$ and $q_2 = f_2$, where $f \in F$ and $f_2 \in F_2$, we obtain that $us'v \Ra u\state{sf_1s_2}v \Ra^* \state{fs_1f_2}$ in $M' \iff usv \Ra^* f$ in $M$, $s_1u \Ra^* f_1$ in $M_1$, and $s_2v \Ra^* f_2$ in $M_2$. Recall that $F' = \{\state{fs_1f_2} \mid f \in F, f_2 \in F_2\}$. Therefore, $L(M') = \{uv \mid usv \Ra^* f \text{ in } M, f \in F, u \in L(M_1), v \in L(M_2)\}$. Since $L(M_1) = A$ and $L(M_2) = B$, $L(M') = \{uv \mid usv \Ra^* f \text{ in } M, f \in F, u \in A, v \in B\}$, so Theorem \ref{theorem:inputABregular_results_in_linear} holds.
\end{proof}

\begin{theorem}\label{theorem:theorem:inputsplitAregular_results_in_linear}
    Let $M = (Q, \Sigma, R, s, F)$ be an $\epsilon$-free \ietwsfa{} and $A \subseteq \Sigma^*$ be regular. Then, there is an \ietwsfa{} $M'$ satisfying $L(M') = \{uv \mid usv \Ra^* f  \text{ in } M, f \in F, uv \in A\}$, so $L(M')$ is linear.
\end{theorem}

\begin{proof}
    Let $M = (Q, \Sigma, R, s, F)$ be an $\epsilon$-free \ietwsfa{}, and let $A \subseteq \Sigma^*$ be regular. Let $A = L(\hat{M})$, where $\hat{M} = (\hat{Q}, \Sigma, \hat{R}, \hat{s}, \hat{F})$ is an $\epsilon$-free \fa{}. From $M$ and $\hat{M}$, we next construct an \ietwsfa{} $M' = (Q', \Sigma, R', s', F')$ such that $L(M') = \{uv \mid usv \Ra^* f \text{ in } M, f \in F, uv \in A\}$. Introduce a new symbol $s'$---the start state of $M'$. Set $\bar{Q} = \{\state{q\hat{p}\hat{q}} \mid q \in Q, \hat{p}, \hat{q} \in \hat{Q}\}$. Assume that $s' \notin \bar{Q}$. Set $Q' = \bar{Q} \cup \{s'\}$ and $F' = \{\langle f\hat{s}\hat{f}\rangle \mid f \in F, \hat{f} \in \hat{F}\}$. Initially, set $R' = \emptyset$. Then, extend $R'$ as follows:
    \begin{enumerate}
        \item for all $\hat{q} \in \hat{Q}$, add $s' \to \state{s\hat{q}\hat{q}}$ to $R'$;\label{theoremproof:theorem:inputsplitAregular_results_in_linear_step_1}
        \item if $ap \to q \in R$ and $\hat{q}a \to \hat{p} \in \hat{R}$, where $p, q \in Q$, $\hat{p}, \hat{q} \in \hat{Q}$, and $a \in \Sigma$, add $a\state{p\hat{p}\hat{o}} \to \state{q\hat{q}\hat{o}}$ to $R'$ for all $\hat{o} \in \hat{Q}$;\label{theoremproof:theorem:inputsplitAregular_results_in_linear_step_2}
        \item if $pa \to q \in R$ and $\hat{p}a \to \hat{q} \in \hat{R}$, where $p, q \in Q$, $\hat{p}, \hat{q} \in \hat{Q}$, and $a \in \Sigma$, add $\state{p\hat{o}\hat{p}}a \to \state{q\hat{o}\hat{q}}$ to $R'$ for all $\hat{o} \in \hat{Q}$.\label{theoremproof:theorem:inputsplitAregular_results_in_linear_step_3}
    \end{enumerate}

    \emph{Basic Idea.} As can be seen, $M'$ works in a two-directional way. To the right, it simulates a~computation made by $M$ and, simultaneously, a computation made by $\hat{M}$ (see step \iref{theoremproof:theorem:inputsplitAregular_results_in_linear_step_3}). To the left, it simulates a computation made by $M$ and, simultaneously, a computation made by $\hat{M}$ in reverse (see step \iref{theoremproof:theorem:inputsplitAregular_results_in_linear_step_2}). Considering step \iref{theoremproof:theorem:inputsplitAregular_results_in_linear_step_1}, observe that $M'$ accepts its input if and only if both $M$ and $\hat{M}$ accept their inputs, too, so $L(M') = \{uv \mid usv \Ra^* f  \text{ in } M, f \in F, uv \in A\}$.

    Complete this proof by analogy with the proof of Theorem \ref{theorem:inputABregular_results_in_linear}.
\end{proof}

\begin{theorem}\label{theorem:inputAregularBfinal_results_in_reg}
    Let $M = (Q, \Sigma, R, s, F)$ be an $\epsilon$-free \ietwsfa{}, $A \subseteq \Sigma^*$ be finite, and $B \subseteq \Sigma^*$ be regular. Then, there exists an \fa{} $M'$ such that $$L(M') = \left\{uv \mid usv \Ra^* f \text{ in } M, f \in F, u \in A, v \in B \right\},$$ so $L(M')$ is regular.
\end{theorem}

\begin{proof}
    Let $M = (Q, \Sigma, R, s, F)$ be an $\epsilon$-free \ietwsfa{}, $A \subseteq \Sigma^*$ be finite, and $B \subseteq \Sigma^*$ be regular. Next, we construct an \fa{} $M' = (Q', \Sigma, R', s', F')$ such that $L(M') = \{uv \mid usv \Ra^* f \text{ in } M, f \in F, u \in A, v \in B\}$. Let $n = \max\{\strlen{x} \mid x \in A\}$. Let $\hat{M} = (\hat{Q}, \Sigma, \hat{R}, \hat{s}, \hat{F})$ be an $\epsilon$-free \fa{} such that $L(\hat{M}) = B$. Set $Q' = \{\state{x}, \state{xq\hat{q}} \mid x \in \Sigma^*, 0 \le \strlen{x} \le n, q \in Q, \hat{q} \in \hat{Q}\}$, $s' = \state{\epsilon}$, and $F' = \{\langle f\hat{f} \rangle \mid f \in F, \hat{f} \in \hat{F}\}$. $R'$ is constructed in the following way:
    \begin{enumerate}
        \item for all $\state{x}, \state{xa} \in Q'$, where $x \in \Sigma^*$ and $a \in \Sigma$, add $\state{x}a \to \state{xa}$ to $R'$;\label{theoremproof:inputAregularBfinal_results_in_reg_step_1}
        \item for all $x \in A$, add $\state{x} \to \state{xs\hat{s}}$ to $R'$;\label{theoremproof:inputAregularBfinal_results_in_reg_step_2}
        \item if $ap \to q \in R$ and $\state{xap\hat{q}}, \state{xq\hat{q}} \in Q'$, where $p, q \in Q$, $\hat{q} \in \hat{Q}$, $a \in \Sigma$ and $x \in \Sigma^*$, add $\state{xap\hat{q}} \to \state{xq\hat{q}}$ to $R'$;\label{theoremproof:inputAregularBfinal_results_in_reg_step_3}
        \item if $pa \to q \in R$, $\hat{p}a \to \hat{q} \in \hat{R}$, and $\state{xp\hat{p}}, \state{xq\hat{q}} \in Q'$, where $q, p \in Q$, $\hat{q}, \hat{p} \in \hat{Q}$, and $x \in \Sigma^*$, add $\state{xp\hat{p}}a \to \state{xq\hat{q}}$ to $R'$.\label{theoremproof:inputAregularBfinal_results_in_reg_step_4}
    \end{enumerate}

    \emph{Basic Idea.} $M'$ starts by reading symbols from its input tape until they form a string from $B$ and records them into its current state (see step \iref{theoremproof:inputAregularBfinal_results_in_reg_step_1}). After this initial phase, $M'$ begins to simulate computations made by $M$ and $\hat{M}$ (see step \iref{theoremproof:inputAregularBfinal_results_in_reg_step_2}). Any left moves made by $M$ are simulated by $M'$ exclusively within its states by successively erasing symbols from the recorded string (see step \iref{theoremproof:inputAregularBfinal_results_in_reg_step_3}). Any right moves made by $M$ and, simultaneously, a~computation made by $\hat{M}$ are simulated by $M'$ by simply reading the remaining part of its input tape (see step \iref{theoremproof:inputAregularBfinal_results_in_reg_step_4}).

    Next, we demonstrate $L(M') = \{uv \mid usv \Ra^* f \text{ in } M, f \in F, u \in A, v \in B\}$ rigorously. We start by proving the following equivalence.

    For all $u, v \in \Sigma^*$, $p, q \in Q$, and $\hat{p}, \hat{q} \in \hat{Q}$ such that $0 \le \strlen{u} \le n$,
    \begin{equation}\label{claim:theorem:inputAregularBfinal_results_in_reg}
        \state{up\hat{p}}v \Ra^* \state{q\hat{q}} \text{ in } M' \iff upv \Ra^* q \text{ in } M \text{ and } \hat{p}v \Ra^* \hat{q} \text{ in } \hat{M}.
    \end{equation}

    First, we establish the \emph{only if} part of this equivalence. By induction on the number of moves $i \ge 0$, we show that $\state{up\hat{p}}v \Ra^i \state{q\hat{q}} \text{ in } M' \implies upv \Ra^* q \text{ in } M \text{ and } \hat{p}v \Ra^* \hat{q} \text{ in } \hat{M}$. Let $i = 0$, so $\state{up\hat{p}}v \Ra^0 \state{q\hat{q}} \text{ in } M'$. Then, $p = q$, $\hat{p} = \hat{q}$, and $uv = \epsilon$. Clearly, $p \Ra^0 p$ in $M$ and $\hat{p} \Ra^0 \hat{p}$ in $\hat{M}$, so the basis holds true. Assume that the implication holds for all computations consisting of no more than $j$ moves in $M'$, for some $j \in \mathbb{N}_0$. Consider any computation of the form $\state{up\hat{p}}v \Ra^{j + 1} \state{q\hat{q}} \text{ in } M'$. Let this computation start with the application of a rule of the form $$\state{u'ap\hat{p}} \to \state{u'o\hat{p}}$$ from $R'$, where $o \in Q$, $u'a = u$, and $a \in \Sigma$. Thus, we can express $\state{up\hat{p}}v \Ra^{j + 1} \state{q\hat{q}}$ as $$\state{u'ap\hat{p}}v \Ra \state{u'o\hat{p}}v \Ra^j \state{q\hat{q}}$$ in $M'$. Since $\state{u'o\hat{p}}v \Ra^j \state{q\hat{q}}$ in $M'$, by the induction hypothesis, $u'ov \Ra^* q$ in $M$ and $\hat{p}v \Ra^* \hat{q}$ in $\hat{M}$. Step \iref{theoremproof:inputAregularBfinal_results_in_reg_step_3} constructs $\state{u'ap\hat{p}} \to \state{u'o\hat{p}} \in R'$ from $ap \to o \in R$, so $$u'apv \Ra u'ov \Ra^* q$$ in $M$. Since $u'a = u$, we have $upv \Ra^* q$ in $M$. Next, suppose that the computation $\state{up\hat{p}}v \Ra^{j + 1} \state{q\hat{q}} \text{ in } M'$ starts with the application of a rule of the form $$\state{up\hat{p}}a \to \state{uo\hat{o}}$$ from $R'$, where $o \in Q$, $\hat{o} \in \hat{Q}$, and $a \in \Sigma$. Express $\state{up\hat{p}}v \Ra^{j + 1} \state{q\hat{q}}$ as $$\state{up\hat{p}}av' \Ra \state{uo\hat{o}}v' \Ra^j \state{q\hat{q}}$$ in $M'$, where $v'a = v$. By the induction hypothesis, $uov' \Ra^* q$ in $M$ and $\hat{o}v' \Ra^* \hat{q}$ in $\hat{M}$. Since step \iref{theoremproof:inputAregularBfinal_results_in_reg_step_4} constructs $\state{up\hat{p}}a \to \state{uo\hat{o}} \in R'$ from $pa \to o \in R$ and $\hat{p}a \to \hat{o} \in \hat{R}$, it follows that $$upav' \Ra uov' \Ra^* q$$ in $M$ and $$\hat{p}av' \Ra \hat{o}v' \Ra^* \hat{q}$$ in $\hat{M}$. Because $av' = v$, $upv \Ra^* q$ in $M$ and $\hat{p}v \Ra^* \hat{q}$ in $\hat{M}$.
    Thus, the induction step is completed.

    Now, we establish the \emph{if} part of equivalence \eqref{claim:theorem:inputAregularBfinal_results_in_reg}, so we show that $upv \Ra^i q \text{ in } M \text{ and } \hat{p}v \Ra^j \hat{q} \text{ in } \hat{M}$, where $j \le i$, implies $\state{up\hat{p}}v \Ra^* \state{q\hat{q}} \text{ in } M'$ by induction on the number of moves $i \ge 0$. Let $i = 0$, so $j = 0$, $pv \Ra^0 \text{ in } M$, and $\hat{p}v \Ra^0 \hat{q} \text{ in } \hat{M}$. Then, $p = q$, $\hat{p} = \hat{q}$, and $uv = \epsilon$. Since $\state{p\hat{p}} \Ra^0 \state{p\hat{p}}$ in $M'$, the basis holds true. Assume that the implication holds for all computations consisting of no more than $k$ moves in $M$, for some $k \in \mathbb{N}_0$. Consider any $upv \Ra^{k + 1} q \text{ in } M \text{ and } \hat{p}v \Ra^l \hat{q} \text{ in } \hat{M}$, where $l \le k$. Let $upv \Ra^{k + 1} q \text{ in } M$ start with the application of a rule of the form $$ap \to o$$ from $R$, where $o \in Q$ and $a \in \Sigma$. Then, express $upv \Ra^{k + 1} q$ as $$u'apv \Ra u'ov \Ra^k q$$ in $M$, where $u = u'a$. Since $u'ov \Ra^k q$ in $M$ and $\hat{p}v \Ra^l \hat{q} \text{ in } \hat{M}$, by the induction hypothesis, $\state{u'o\hat{p}}v \Ra^* \state{q\hat{q}}$ in $M'$. From $ap \to o \in R$, step \iref{theoremproof:inputAregularBfinal_results_in_reg_step_3} constructs $\state{u'ap\hat{p}} \to \state{u'o\hat{p}} \in R'$, so $$\state{u'ap\hat{p}}v \Ra \state{u'o\hat{p}}v \Ra^* \state{q\hat{q}}$$ in $M'$. Because $u'a = u$, $\state{up\hat{p}}v \Ra^* \state{q\hat{q}}$ in $M'$. Next, consider any $upv \Ra^{k + 1} q \text{ in } M \text{ and } \hat{p}v \Ra^{l + 1} \hat{q} \text{ in } \hat{M}$ with $l \le k$. Let $upv \Ra^{k + 1} q \text{ in } M$ start with the application of a rule of the form $$pa \to o$$ from $R$ and $\hat{p}v \Ra^{l + 1} \hat{q} \text{ in } \hat{M}$ start with the application of a rule of the form $$\hat{p}a \to \hat{o}$$ from $\hat{R}$, where $o \in Q$, $\hat{o} \in \hat{Q}$, and $a \in \Sigma$. Express $upv \Ra^{k + 1} q$ as $$upav' \Ra uov' \Ra^k q$$ in $M$ and $\hat{p}v \Ra^{l + 1} \hat{q}$ as $$\hat{p}av' \Ra \hat{o}v' \Ra^l \hat{q}$$ in $\hat{M}$, where $av' = v$. By the induction hypothesis, $\state{uo\hat{o}}v' \Ra^* \state{q\hat{q}}$ in $M'$. From $pa \to o \in R$ and $\hat{p}a \to \hat{o} \in \hat{R}$, step \iref{theoremproof:inputAregularBfinal_results_in_reg_step_4} constructs $\state{up\hat{p}}a \to \state{uo\hat{o}} \in R'$, so $$\state{up\hat{p}}av' \Ra \state{uo\hat{o}}v' \Ra^* \state{q\hat{q}}$$ in $M'$. Since $av' = v$, $\state{up\hat{p}}v \Ra^* \state{q\hat{q}}$ in $M'$. Thus, the induction step is completed, and equivalence \eqref{claim:theorem:inputAregularBfinal_results_in_reg} holds.

    Considering equivalence \eqref{claim:theorem:inputAregularBfinal_results_in_reg} for $p = s$ and $p' = s'$, we see that for all $u, v \in \Sigma^*$, $q \in Q$, and $\hat{q} \in \hat{Q}$ such that $0 \le \strlen{u} \le n$, $\state{us\hat{s}}v \Ra^* \state{q\hat{q}} \text{ in } M' \iff usv \Ra^* q \text{ in } M \text{ and } \hat{s}v \Ra^* \hat{q} \text{ in } \hat{M}$. As follows from the construction of $R'$, $M'$ starts every accepting computation by a sequence of moves of the form $\state{\epsilon} \Ra^* \state{x} \Ra \state{xs\hat{s}}$, where $x \in A$. Consequently, $\state{\epsilon}v \Ra^* \state{u}v \Ra \state{us\hat{s}}v \Ra^* \state{q\hat{q}} \text{ in } M' \iff usv \Ra^* q \text{ in } M \text{, } \hat{s}v \Ra^* \hat{q} \text{ in } \hat{M}$, and $u \in A$. Now, consider this equivalence for $q = f$ and $\hat{q} = \hat{f}$, where $f \in F$ and $\hat{f} \in \hat{F}$. That is, $\state{\epsilon}v \Ra^* \state{u}v \Ra \state{us\hat{s}}v \Ra^* \langle f\hat{f} \rangle \text{ in } M' \iff usv \Ra^* f \text{ in } M \text{, } \hat{s}v \Ra^* \hat{f} \text{ in } \hat{M}$, and $u \in A$. Since $F' = \{\langle f\hat{f} \rangle \mid f \in F, \hat{f} \in \hat{F}\}$, it follows that $L(M') = \{uv \mid usv \Ra^* f \text{ in } M, f \in F, u \in A, v \in L(\hat{M})\}$. Given that $L(\hat{M}) = B$, we have $L(M') = \{uv \mid usv \Ra^* f \text{ in } M, f \in F, u \in A, v \in B\}$. Therefore, Theorem \ref{theorem:inputAregularBfinal_results_in_reg} holds.
\end{proof}

\begin{theorem}\label{theorem:inputAregBregCreg_readingB_results_in_reg}
    Let $M = (Q, \Sigma, R, s, F)$ be an $\epsilon$-free \ietwsfa{}, and let $A, B, C \subseteq \Sigma^*$ be regular. Then, there exists an \fa{} $M'$ such that $L(M') = \{v \mid usvw \Ra^* f  \text{ in } M, f \in F, u \in A, v \in B, w \in C\}$, so $L(M')$ is regular.
\end{theorem}

\begin{proof}
    Let $M = (Q, \Sigma, R, s, F)$ be an $\epsilon$-free \ietwsfa{} and $A, B, C \subseteq \Sigma^*$ be regular. Let $A = L(M_1)$, $B = L(M_2)$, and $C = L(M_3)$, where $M_i = (Q_i, \Sigma_i, R_i, s_i, F_i)$ is an $\epsilon$-free \fa{} for all $i \in \{1, 2, 3\}$. From $M$, $M_1$, $M_2$, and $M_3$, we construct an $\epsilon$-free \fa{} $M' = (Q', \Sigma, R', s', F')$ satisfying $L(M') = \left\{v \mid usvw \Ra^* f \text{ in } M, f \in F, u \in A, v \in B, w \in C\right\}$. Introduce a new symbol $s'$---the start state of $M'$. Set $\hat{Q} = \{\state{qq_1q_2s_31}, \state{qq_1f_2q_32} \mid q \in Q, q_i \in Q_i, i \in \{1, 2, 3\}, f_2 \in F_2\}$. Without any loss of generality, assume that $s' \notin \hat{Q}$. Set $Q' = \hat{Q} \cup \{s'\}$ and $F' = \{\state{fs_1f_2f_32} \mid f \in F, f_i \in F_i, i \in \{2, 3\}\}$. $R'$ is constructed as follows.
    \begin{enumerate}
        \item for all $f_1 \in F_1$, add $s' \to \state{sf_1s_2s_31}$ to $R'$;\label{theoremproof:inputAregBregCreg_readingB_results_in_reg_step_1}
        \item for all $q \in Q$, $q_1 \in Q_1$, $f_2 \in F_2$, add $\state{qq_1f_2s_31} \to \state{qq_1f_2s_32}$ to $R'$;\label{theoremproof:inputAregBregCreg_readingB_results_in_reg_step_2}
        \item if $ap \to q \in R$ and $q_1a \to p_1 \in R_1$, where $p, q \in Q$, $p_1, q_1 \in Q_1$, and $a \in \Sigma$, then add $\state{pp_1q_2s_31} \to \state{qq_1q_2s_31}$ and $\state{pp_1f_2q_32} \to \state{qq_1f_2q_32}$ to $R'$ for all $q_2 \in Q_2$, $q_3 \in Q_3$, and $f_2 \in F_2$;\label{theoremproof:inputAregBregCreg_readingB_results_in_reg_step_3}
        \item if $pa \to q \in R$ and $p_2a \to q_2 \in R_2$, where $p, q \in Q$, $p_2, q_2 \in Q_2$, and $a \in \Sigma$, then add $\state{pq_1p_2s_31}a \to \state{qq_1q_2s_31}$ to $R'$ for all $q_1 \in Q_1$;\label{theoremproof:inputAregBregCreg_readingB_results_in_reg_step_4}
        \item for each $pa \to q \in R$ and $p_3a \to q_3 \in R_3$, where $p, q \in Q$, $p_3, q_3 \in Q_3$, and $a \in \Sigma$, add $\state{pq_1f_2p_32} \to \state{qq_1f_2q_32}$ to $R'$ for all $q_1 \in Q_1$ and $f_2 \in F_2$.\label{theoremproof:inputAregBregCreg_readingB_results_in_reg_step_5}
    \end{enumerate}

    To establish $L(M') = \left\{v \mid usvw \Ra^* f  \text{ in } M, f \in F, u \in A, v \in B, w \in C\right\}$ formally, we first prove equivalences \eqref{claim:theorem:inputAregBregCreg_readingB_results_in_reg_A} and \eqref{claim:theorem:inputAregBregCreg_readingB_results_in_reg_B}, given next.

    For all $v \in \Sigma^*$, $p, q \in Q$, $p_1, q_1 \in Q_1$, and $p_2, q_2 \in Q_2$
    \begin{equation}\label{claim:theorem:inputAregBregCreg_readingB_results_in_reg_A}
        \state{pp_1p_2s_31}v \Ra^* \state{qq_1q_2s_31} \text{ in } M' \iff \text{there is $x \in \Sigma^*$ such that } \begin{cases}
            xpv \Ra^* q \text{ in } M, \\
            q_1x \Ra^* p_1 \text{ in } M_1, \text{ and} \\
            p_2v \Ra^* q_2 \text{ in } M_2.
        \end{cases}
    \end{equation}

    First, we establish the \emph{only if} part of equivalence \eqref{claim:theorem:inputAregBregCreg_readingB_results_in_reg_A}. By induction on the number of moves $i \ge 0$, we show that $\state{pp_1p_2s_31}v \Ra^* \state{qq_1q_2s_31} \text{ in } M'$ implies that there is $x \in \Sigma^*$ such that $xpv \Ra^* q \text{ in } M$, $q_1x \Ra^* p_1 \text{ in } M_1$, and $p_2v \Ra^* q_2 \text{ in } M_2$. Let $i = 0$, so $\state{pp_1p_2s_31}v \Ra^0 \state{qq_1q_2s_31} \text{ in } M'$. Then, $p = q$, $p_1 = q_1$, $p_2 = q_2$, and $v = \epsilon$. Clearly, $p \Ra^0 p$ in $M$, $p_1 \Ra^0 p_1$ in $M_1$, and $p_2 \Ra^0 p_2$ in $M_2$, so the basis holds true. Assume that the implication holds for all computations consisting of no more than $j$ moves in $M'$, for some $j \in \mathbb{N}_0$. Consider any computation of the form $\state{pp_1p_2s_31}v \Ra^{j + 1} \state{qq_1q_2s_31} \text{ in } M'$. Let this computation start with the application of a rule of the form $$\state{pp_1p_2s_31} \to \state{oo_1p_2s_31}$$ from $R'$, where $o \in Q$ and $o_1 \in Q_1$. Now, express $\state{pp_1p_2s_31}v \Ra^{j + 1} \state{qq_1q_2s_31}$ as $$\state{pp_1p_2s_31}v \Ra \state{oo_1p_2s_31}v \Ra^j \state{qq_1q_2s_31}$$ in $M'$. Since $\state{oo_1p_2s_31}v \Ra^j \state{qq_1q_2s_31}$ in $M'$, by the induction hypothesis, $x'ov \Ra^* q$ in $M$, $q_1x' \Ra^* o_1$ in $M_1$, and $p_2v \Ra^* q_2$ in $M_2$, for some $x' \in \Sigma^*$. Step \iref{theoremproof:inputAregBregCreg_readingB_results_in_reg_step_3} constructs $\state{pp_1p_2s_31} \to \state{oo_1p_2s_31} \in R'$ from $ap \to o \in R$ and $o_1a \to p_1 \in R_1$, for some $a \in \Sigma$, so $$x'apv \Ra x'ov \Ra^* q$$ in $M$ and $$q_1x'a \Ra^* o_1a \Ra p_1$$ in $M_1$. Hence, assuming that $x = x'a$, we have $xpv \Ra^* q$ in $M$ and $q_1x \Ra^* p_1$ in $M_1$. If the computation $\state{pp_1p_2s_31}v \Ra^{j + 1} \state{qq_1q_2s_31} \text{ in } M'$ starts with the application of a~rule of the form $\state{pp_1p_2s_31}a \to \state{op_1o_2s_31}$ from $R'$, where $o \in Q$, $o_2 \in Q_2$, and $a \in \Sigma$, proceed analogously. Thus, the induction step is completed.

    Next, we establish the \emph{if} part of the equivalence \eqref{claim:theorem:inputAregBregCreg_readingB_results_in_reg_A}. By induction on the number of moves $i \ge 0$, we prove that $xpv \Ra^i q \text{ in } M$, $q_1x \Ra^j p_1 \text{ in } M_1$, and $p_2v \Ra^k q_2 \text{ in } M_2$, where $j + k = i$, implies $\state{pp_1p_2s_31}v \Ra^* \state{qq_1q_2s_31} \text{ in } M'$. Let $i = 0$, so $j = 0$, $k = 0$, $xpv \Ra^0 q \text{ in } M$, $q_1x \Ra^0 p_1 \text{ in } M_1$, and $p_2v \Ra^0 q_2 \text{ in } M_2$. Then, $p = q$, $p_1 = q_1$, $p_2 = q_2$, and $xv = \epsilon$. Since $\state{pp_1p_2s_31}v \Ra^0 \state{pp_1p_2s_31} \text{ in } M'$, the basis holds true. Assume that the implication holds for all computations consisting of no more than $l$ moves in $M$, for some $l \in \mathbb{N}_0$. Consider any $xpv \Ra^{l + 1} q \text{ in } M$, $q_1x \Ra^{m + 1} p_1 \text{ in } M_1$, and $p_2v \Ra^{n} q_2 \text{ in } M_2$, where $m + n = l$. Let $xpv \Ra^{l + 1} q \text{ in } M$ start with the application of the form $$ap \to o$$ from $R$ and $q_1x \Ra^{m + 1} p_1 \text{ in } M_1$ end with the application of a rule of the form $$o_1a \to p_1$$ from $R_1$, where $o \in Q$, $o_1 \in Q_1$, and $a \in \Sigma$. Express $xpv \Ra^{l + 1} q$ as $$x'apv \Ra x'ov \Ra^l q$$ in $M$ and $q_1x \Ra^{m + 1} p_1$ as $$q_1x'a \Ra^m o_1a \Ra p_1$$ in $M_1$, where $x'a = x$. By the induction hypothesis, $\state{oo_1p_2s_31}v \Ra^* \state{qq_1q_2s_31}$ in $M'$. From $ap \to o \in R$ and $o_1a \to p_1 \in R_1$, step \iref{theoremproof:inputAregBregCreg_readingB_results_in_reg_step_3} constructs $\state{pp_1p_2s_31} \to \state{oo_1p_2s_31} \in R'$, so $$\state{pp_1p_2s_31}v \Ra \state{oo_1p_2s_31}v \Ra^* \state{qq_1q_2s_31}$$ in $M'$. Therefore, $\state{pp_1p_2s_31}v \Ra^* \state{qq_1q_2s_31}$ in $M'$. Next, consider any $xpv \Ra^{l + 1} q \text{ in } M$, $q_1x \Ra^{m} p_1 \text{ in } M_1$, and $p_2v \Ra^{n + 1} q_2 \text{ in } M_2$, where $m + n = l$. Let $xpv \Ra^{l + 1} q \text{ in } M$ start with the application of a rule of the form $pa \to o$ from $R$ and $p_2v \Ra^{n + 1} q_2 \text{ in } M_2$ start $p_2a \to o_2$ from $R_2$, where $o \in Q$, $o_2 \in Q_2$, and $a \in \Sigma$. Then, proceed by analogy with the previous case. Thus, the induction step is completed, and equivalence \eqref{claim:theorem:inputAregBregCreg_readingB_results_in_reg_A} holds.

    For all $q, o \in Q$, $q_1, o_1 \in Q_1$, $f_2 \in F_2$, and $q_3, o_3 \in Q_3$
    \begin{equation}\label{claim:theorem:inputAregBregCreg_readingB_results_in_reg_B}
        \state{qq_1f_2q_32} \Ra^* \state{oo_1f_2o_32} \text{ in } M' \iff \text{there is $y, w \in \Sigma^*$ such that } \begin{cases}
            yqw \Ra^* o \text{ in } M, \\
            o_1y \Ra^* q_1 \text{ in } M_1, \text{ and} \\
            q_3w \Ra^* o_3 \text{ in } M_3.
        \end{cases}
    \end{equation}

    Prove equivalence \eqref{claim:theorem:inputAregBregCreg_readingB_results_in_reg_B} by analogy with the proof of equivalence \eqref{claim:theorem:inputAregBregCreg_readingB_results_in_reg_A}.

    Observe that $M'$ starts every accepting computation by using a rule of the form $s' \to \state{sf_1s_2s_31}$, where $f_1 \in F_1$, and also applies a rule of the form $\state{qq_1f_2s_31} \to \state{qq_1f_2s_32}$, where $q \in Q$, $q_1 \in Q_1$, and $f_2 \in F_2$, at some point during each such computation (see steps~\iref{theoremproof:inputAregBregCreg_readingB_results_in_reg_step_1} and~\iref{theoremproof:inputAregBregCreg_readingB_results_in_reg_step_2}). From these observations and equivalences \eqref{claim:theorem:inputAregBregCreg_readingB_results_in_reg_A} and \eqref{claim:theorem:inputAregBregCreg_readingB_results_in_reg_B}, it follows that for all $q, o \in Q$, $q_1, o_1 \in Q_1$, $o_3 \in Q_3$, $f_1 \in F_1$, $f_2 \in F_2$, and $v \in \Sigma^*$, $s'v \Ra \state{sf_1s_2s_31}v \Ra^* \state{qq_1f_2s_31} \Ra \state{qq_1f_2s_32} \Ra^* \state{os_1f_2o_32}$ in $M'$ iff there is $u, w \in \Sigma^*$ such that $usvw \Ra^* o$ in $M$, $s_1u \Ra^* f_1$ in $M_1$, $s_2v \Ra^* f_2$ in $M_2$, and $s_3w \Ra^* o_3$ in $M_3$. Considering this equivalence for $o = f$ and $o_3 = f_3$, where $f \in F$ and $f_3 \in F_3$, we obtain that $s'v \Ra \state{sf_1s_2s_31}v \Ra^* \state{qq_1f_2s_31} \Ra \state{qq_1f_2s_32} \Ra^* \state{fs_1f_2f_32}$ in $M'$ iff there is $u, w \in \Sigma^*$ such that $usvw \Ra^* f$ in $M$, $s_1u \Ra^* f_1$ in $M_1$, $s_2v \Ra^* f_2$ in $M_2$, and $s_3w \Ra^* f_3$ in $M_3$. Recall that $F' = \{\state{fs_1f_2f_32} \mid f \in F, f_i \in F_i, i \in \{2, 3\}\}$. Hence, $L(M') = \{v \mid usvw \Ra^* f  \text{ in } M, f \in F, u \in L(M_1), v \in L(M_2), w \in L(M_3)\}$. As $L(M_1) = A$, $L(M_2) = B$, and $L(M_3) = C$, we have $L(M') = \{v \mid usvw \Ra^* f \text{ in } M, f \in F, u \in A, v \in B, w \in C\}$. Therefore, Theorem \ref{theorem:inputAregBregCreg_readingB_results_in_reg} holds.
\end{proof}

\section{Conclusion}\label{sec:conclusion}

The present paper has proposed new versions of two-way finite automata referred to as input-erasing two-way finite automata (see Section \ref{sec:definitions}). In essence, they perform their computation just like the classical versions of these automata (see \cite{FINITE_AUTOMATA_DECISION_PROBLEMS_RABIN_1959, TWO_WAY_FINITE_AUTOMATA_SHEPHERDSON_1959}) except that 
\begin{enumerate*}
    \item they erase the input symbols just like one-way finite automata do, and
    \item they start their computation at any position on the input tape.
\end{enumerate*}
Although this paper has established several fundamental results concerning these new automata and their restricted versions (see Sections \ref{sec:main_result} through \ref{sec:input_related_restrictions}), there still remain several open problem areas to study, including
\begin{enumerate}[label=(\roman*)]
    \item an investigation of more classical topics of automata theory, such as determinism and minimization;
    \item a further investigation of input-related restrictions, for instance, in terms of subregular language families;
    \item a conceptualization and investigation of input-erasing two-way finite automata in an alternative way by analogy with other modern concepts of automata, such as regulated and jumping versions (see \cite{REGULATED_GRAMMARS_AND_AUTOMATA_MEDUNA_2014, JUMPING_COMPUTATION_MEDUNA_2024}); and
    \item an introduction and investigation of other types of automata, such as pushdown automata (see \cite{INTRODUCTION_TO_AUTOMATA_THEORY_HOPCROFT_1979}), conceptualized by analogy with the input-erasing two-way finite automata given in the present paper.
\end{enumerate}

\subsection*{Acknowledgements}

This work was supported by the BUT grant FIT-S-23-8209.

\bibliographystyle{fundam}
\bibliography{references}

\end{document}